\documentclass{new_tlp}
\usepackage{mathptmx}
\usepackage[english]{babel}
\usepackage{caption}
\usepackage{stmaryrd}
\usepackage{epsfig}
\usepackage{bbm}
\usepackage{latexsym}
\usepackage{times}
\usepackage{helvet}
\usepackage{courier}
\usepackage{graphicx}
\usepackage{wrapfig}
\usepackage{float}
\usepackage{mathrsfs}
\usepackage{dsfont}
\usepackage{amsmath}
\usepackage{amssymb}
\usepackage{verbatim}
\usepackage{url}
\usepackage{boxedminipage}
\usepackage{pifont}
\usepackage{xspace}
\usepackage{color}
\usepackage[tikz]{bclogo}
\usepackage{ifthen}
\usepackage{todonotes}
\usepackage{mathdots}
\usepackage{algorithm}
\usepackage[noend]{algpseudocode}
\usepackage{txfonts}
\usepackage{mathtools}
\usepackage{relsize}
\usepackage{booktabs}
\usepackage{multirow}
\usepackage{extarrows}
\usepackage{enumerate}

\newcommand\bcmdtab{\noindent\bgroup\tabcolsep=0pt%
  \begin{tabular}{@{}p{10pc}@{}p{20pc}@{}}}
\newcommand\ecmdtab{\end{tabular}\egroup}

  \title[Theory and Practice of Logic Programming]
        {Restricted Chase Termination for Existential Rules: a Hierarchical Approach and Experimentation}

  \author[A. Karimi, H. Zhang, J-H. You]
         {Arash Karimi{$^1$}, Heng Zhang{$^2$}, Jia-Huai You{$^1$}
         	\\
         {$^1$}Department of Computing Science, University of Alberta, Edmonton, Canada\\{$^2$}School of Software Engineering, Tianjin University, 
         Tianjin, 
         China\\
         \email{akarimi@ualberta.ca}}

\jdate{March 2003}
\pubyear{2003}
\pagerange{\pageref{firstpage}--\pageref{lastpage}}
\doi{S1471068401001193}

\algnewcommand{\algorithmicgoto}{\textbf{goto}}%
\algnewcommand{\Goto}[1]{\algorithmicgoto~\ref{#1}}%

\newcommand{\WA}{\text{WA}}
\newcommand{\aGRD}{\text{aGRD}}
\newcommand{\JA}{\text{JA}}
\newcommand{\SWA}{\text{SWA}}

\newcommand{\kSAFE}{\text{$k$-$\mathsf{safe}$}}

\newtheorem{thm}{Theorem}
\newtheorem{cor}[thm]{Corollary}

\newtheorem{prop}[thm]{Proposition}

\newtheorem{defn}{Definition}
\newtheorem{rem}{Remark}
\newtheorem{exm}{Example}

\DeclareMathAlphabet{\mathbfit}{OML}{cmm}{b}{it}

\long\def\comment#1{}

\makeatletter
\def\BState{\State\hskip-\ALG@thistlm}
\makeatother

\begin{document}

\maketitle

  \begin{abstract}
    The chase procedure for existential rules is an indispensable tool for several database applications, where its termination guarantees the decidability of these tasks. Most previous studies have focused on the skolem chase variant and its termination analysis. It is known that the restricted chase variant is a more powerful tool in termination analysis provided a database is given. But all-instance termination presents a challenge since the critical database and similar techniques do not work.
    In this paper, we develop a novel technique to characterize the activeness of all possible cycles of a certain length for the restricted chase, which leads to the formulation of a parameterized class of the finite restricted chase, called $k$-$\mathsf{safe}(\Phi)$.
    This approach applies to any class of finite skolem chase identified with a condition of acyclicity. 
    More generally, we show that the approach can be applied to the hierarchy of {\em bounded rule sets} previously only defined for the skolem chase.
   Experiments on a collection of ontologies from the web show the applicability of the proposed methods on real-world ontologies. Under consideration in Theory and Practice of Logic Programming (TPLP).
  \end{abstract}

  \begin{keywords}
    Existential Rules, Ontological Reasoning, Termination Analysis, Complexity of Reasoning
  \end{keywords}

\section{Introduction}

The advent of emerging applications of knowledge representation and ontological reasoning has been the motivation of
recent studies on rule-based languages, known as tuple-generating dependencies (TGDs) \cite{beeri1984proof}, existential rules \cite{baget2011rules} or Datalog$^{\pm}$ \cite{cali-lics10}, which have been considered as a powerful modeling language for applications in data exchange, data integration, ontological querying, and so on. A major advantage of this approach is that the formal semantics based on first-order logic facilitates reasoning in an application, where answering conjunctive queries over a database extended with a set of existential rules is a primary task, but unfortunately an undecidable one in general \cite{BV81}. The {\em chase procedure} is a bottom-up algorithm that extends a given database by applying specified rules. 
If such a procedure terminates, given an input database $I$, a finite rule set $R$ and a conjunctive query, we can answer the query against $R$ and $I$ by simply evaluating it on the result of the chase.
In applications such as in data exchange scenarios, we need the result that the chase terminates for all databases. Thus, determining if the chase of a rule set terminates is crucial in these applications.

Existential rules in this context are implications of the form ${\forall {\mathbf{x}}\forall{\mathbf{y}}~ (\phi({\mathbf{x}},{\mathbf{y}})\rightarrow\exists{\mathbf{z}} ~\psi({{\mathbf{x},\mathbf{z}})})}$, where $\phi$ and $\psi$ are conjunctions of atoms. 

For example, that {\em every student has a classmate who is also a student} can be expressed by 
$$\mathrm{Student}(x) \rightarrow \exists z ~ \mathrm{Classmate}(x,z), \mathrm{Student}(z)$$
where universal quantifiers are omitted.

We can remove existential quantifiers by skolemization where existential variables are replaced by skolem terms.
For the above example, the resulting skolemized rule is $$\mathrm{Student}(x) \rightarrow \mathrm{Classmate}(x,f_z(x)), \mathrm{Student}(f_z(x))$$
Given a database, say $I = \{\mathrm{Student}(a)\}$, the atom in it triggers the application of the rule, which will first add in $I$ the atoms $\mathrm{Classmate}(a,f_z(a))$, $\mathrm{Student}(f_z(a))$; repeated applications will further add $\mathrm{Classmate}(f_z(a), f_z(f_z(a)))$, $\mathrm{Student}(f_z(f_z(a)))$, and so on. In this example, the chase produces an infinite set.

Note that a set of skolemized rules is a Horn logic program.

Four main variants of the chase procedure have been considered in the literature, which are called {\em oblivious} \cite{fagin2005data}, {\em skolem} \cite{marnette2009generalized} ({\em semi-oblivious}),\footnote{The chase using skolemized rules can be expressed equivalently by introducing fresh nulls. The chase under these two different notations are considered equivalent due to a one-to-one correspondence between generated skolem terms and introduced fresh nulls.} {\em restricted} (a.k.a. {\em standard}) \cite{fagin2005data} and the {\em core} chase \cite{deutsch2008chase}.

What is common to all these chase variants is the property that, for any database instance $I$, a finite rule set $R$ and a Boolean conjunctive query $q$, $q$ is entailed by $R$ and $I$ if and only if it is entailed by the result of the chase on $R$ and $I$. However, they behave differently concerning termination. The oblivious chase is weaker than the skolem chase, in the sense that whenever the oblivious chase terminates, so does the skolem chase, but the reverse does not hold in general. In turn, the skolem chase is weaker than the restricted chase, which is itself weaker than the core chase.

The core chase is defined based on the restricted chase combined with the notion of {\em cores of relational structures} \cite{hell1992core}. This chase variant is theoretically interesting as it captures all universal models of a given rule set and instance.\footnote{Given an instance $I$ and a rule set $R$, an instance $J$ is a model of $R$ and $I$ if $J$ satisfies all rules in $R$ and there is a homomorphism from $I$ to $J$. Moreover, a model $U$ is universal for $R$ and $I$ if it has homomorphism into every model of $R$ and $I$. Models of $R$ and $I$ are not unique, but universal models of $R$ and $I$ are unique up to homomorphism.}
Given a rule set $R$ and an instance $I$, whenever there is a universal model of $R$ and $I$, the core chase produces the smallest such model.

As the cost of each step of the core chase is DP-complete, this chase variant is a bit more complicated than the other main chase variants and to the best of our knowledge, there are no known efficient algorithms to compute the core when the instances under evaluation are of non-trivial sizes.

In this paper, we focus on the skolem and the restricted versions of the chase, which have been the most investigated in the literature, as the core chase is computationally costly in practice (cf. \cite{benedikt2017benchmarking} for more details).

%In particular, given a rule set and a database, we know that whenever the skolem chase terminates so does the restricted chase, but the reverse does not hold in general.

Despite the existence of many notions of acyclicity in the literature (cf. \cite{cuenca2013acyclicity} for a survey), there are natural examples from real-world ontologies that are non-terminating under the skolem chase but terminating under the restricted chase. However, finding a suitable characterization to ensure restricted chase termination is a challenging task, and in the last decade, to the best of our knowledge, only a few conditions have been discovered.
In \cite{carral2017restricted}, the classes of {\em restricted joint acyclicity} (RJA), {\em restricted model-faithful acyclicity} (RMFA) and {\em restricted model-summarizing acyclicity} (RMSA) of finite all-instance, all-path restricted chase are introduced which generalize the corresponding classes under the skolem chase, namely (by removing the letter R in the above names) {\em joint-acyclicity} (JA) \cite{krotzsch2011extending}, {\em model-faithful acyclicity} (MFA) and {\em model-summarizing acyclicity} (MSA) \cite{cuenca2013acyclicity}, respectively. Intuitively, the classes for the restricted chase introduce a {\em blocking criterion} to check if the head of each rule is already entailed by the derivations when constructing the arena for checking the corresponding acyclicity conditions for JA, MFA, and  MSA, respectively.
Here, we extend their work in two different directions.
First, we provide a highly general theoretical framework to identify strict superclasses of all existing classes of finite skolem chase that we are aware of, and second, we show a general critical database technique, which works uniformly for all bounded finite chase classes.

With the curiosity on the intended applications of some of the practical ontologies that we collected from the web (which will be used in our experimentation to be reported later in this paper) and the question why the restricted chase may help identify classes of terminating rule sets, 
we analyze some of them to get an understanding. Here, let us introduce a case-study of policy analysis for access control, which is abstracted from a practical ontology from the considered collection. 
This example shows how the user may utilize the approach we have developed in this paper to model and reason with a particular access policy.

Consider a scenario involving several research groups in a given lab located in a department.
Each one of these groups may have some personnel working in labs. Also, each person may possess keys which are access cards to the labs of that department.
The set of rules $R=\{r_1,r_2,r_3,r_4,r_5\}$ below is intended to model the access policy to the labs:
any member of any research lab must be able to enter their lab that is assigned to the research group ($r_1$);
for each person $x$ who has a key to a room $y$ there is a lab $u$ such that $x$ can enter $u$ and the key $y$ opens the door of that lab ($r_2$);
and if a person can enter a lab, he or she must have a matching key that opens the lab ($r_3$).

An employee of the department is responsible for granting the keys to labs ($r_4$). Once an employee grants a key to a person, the grantee is assumed to be in the possession of the key ($r_5$). 

$$\begin{array}{ll}
\!\!r_1:  \mathrm{MemOf}(x,y) \rightarrow \mathrm{Enters}(x,y)\\
\!\!r_2:  \mathrm{HasKey}(x,y) \rightarrow \exists u\ \mathrm{Enters}(x,u), \mathrm{KeyOpens}(y,u)\\
	\!\!r_3:  \mathrm{Enters}(x,y) \rightarrow \exists v\ \mathrm{HasKey}(x,v), \mathrm{KeyOpens}(v,y)\\
	\!\!r_4:  \mathrm{HasKey}(x,y) \rightarrow \exists w\ \mathrm{Grants}(w,x,y), \mathrm{Emp}(w)\\
	\!\!r_5:  \mathrm{Grants}(t,x,y) \rightarrow \mathrm{HasKey}(x,y)
	%\\
	%\!\!r_6:  LocatedIn(Robotics, \mathsf{CSC}232)\\
    %\!\!r_7:  LocatedIn(AI, \mathsf{CSC}232)
\end{array}
$$

The intended meanings of the predicates are: 
$\mathrm{MemOf}(x,y)$ represents that $x$ is a member of (lab) $y$;
$\mathrm{Enters}(x,y)$ says that (person) $x$ enters (lab) $y$; $\mathrm{HasKey}(x,y)$ affirms that (person) $x$ has a key card to (room) $y$; $\mathrm{KeyOpens}(y,u)$ means that the key to (room) $y$ opens (lab) $u$; Furthermore, by $\mathrm{Grants}(w,x,y)$, we declare that (employee) $w$ grants (person) $x$ access to (room) $y$; Finally, $\mathrm{Emp}(w)$ confirms that $w$ is an employee of the department.

The rules in $R$ can be applied cyclically.
For example, an application of $r_4$ triggers an application of $r_5$ which triggers $r_4$ again. But even under the skolem chase variant, these two rules do not produce an infinite derivation sequence.
Let us consider the path $\pi_1=(r_4,r_5)$ and show the skolem chase derivations of $sk(\pi_1)$ from $\{\mathrm{HasKey}\big(a,b\big)\}$.
Recall that the skolem chase considers the skolemized version of the rules.
	$$
	\begin{array}{ll}
	I_0 = \{\mathrm{HasKey}\big(a,b\big)\} \xlongrightarrow{\langle sk(r_4),\{x/a, y/b\}\rangle}\\
	I_1 = I_0 \cup \big\{\mathrm{Grants}\big(f_w(a,b),a,b\big), \mathrm{Emp}\big(f_w(a,b)\big)\big\} \xlongrightarrow{\langle sk(r_5), \{t/f_w(a,b),x/a,y/b\}\rangle}\\
	I_2 = I_1 
	%\cup \big\{HasKey\big(a,b\big)\big\}
	%\\\xlongrightarrow{\langle sk(r_2), \{x/a, y/f_v(a,f_u(a,b))\}\rangle}\\I_3 = I_2 \cup \big\{\mathrm{Enters}\big(a, f_u(a,f_v(a,f_u(a,b)))\big), \mathrm{KeyOpens}\big(f_v(a,f_u(a,b)),f_u(a,f_v(a,f_u(a,b)))\big)\}\\\dots
	\end{array}
	$$
where $\xrightarrow{\langle sk(r), \tau\rangle}$ denotes that rule $sk(r)$ is applied using substitution $\tau$.

The sequence of derivations for the path $\pi_2=(r_5,r_4)$ can be obtained similarly.
From these derivations, we can observe that any path of rules that only consist of $r_4$ and $r_5$ is terminating under the skolem chase.

However, the cyclic applications of $r_2$ and $r_3$ lead to an infinite skolem chase. 
To illustrate, let us construct a skolem chase sequence starting from the application of $r_2$ on a singleton database $\{\mathrm{HasKey}(a,b)\}$ as follows (where the existential variable $u$ in $r_2$ is skolemized to $f_u(x,y)$ and $v$ in $r_3$ is skolemized to $f_v(x,y)$):
	$$
	\begin{array}{ll}
	I_0 = \{\mathrm{HasKey}\big(a,b\big)\} \xlongrightarrow{\langle sk(r_2),\{x/a, y/b\}\rangle}\\
	I_1 = I_0 \cup \big\{\mathrm{Enters}\big(a, f_u(a,b)\big), \mathrm{KeyOpens}\big(b,f_u(a,b)\big)\big\} \xlongrightarrow{\langle sk(r_3), \{x/a, y/f_u(a,b)\}\rangle}\\
	I_2 = I_1 \cup \big\{\mathrm{HasKey}\big(a, f_v(a,f_u(a,b))\big), \mathrm{KeyOpens}\big(f_v(a,f_u(a,b)),f_u(a,b)\big)\big\}\\
	\xlongrightarrow{\langle sk(r_2), \{x/a, y/f_v(a,f_u(a,b))\}\rangle}\\
	I_3 = I_2 \cup \big\{\mathrm{Enters}\big(a, f_u(a,f_v(a,f_u(a,b)))\big), \mathrm{KeyOpens}\big(f_v(a,f_u(a,b)),f_u(a,f_v(a,f_u(a,b)))\big)\}\\
	\dots
	\end{array}
	$$

	On the other hand, in each valid derivation of a restricted chase sequence, we must ensure that each rule $r_i$ that is used in the derivation is not already satisfied by the current conclusion set, which is the set of all derivations generated so far right before application of $r_i$.

Though the skolem chase leads to an infinite sequence, the restricted chase does terminate. 
Utilizing fresh nulls, denoted by $n_i$, for the representation of unknowns,\footnote{For the clarity of illustration, we use fresh nulls instead of skolem terms \-- there is a one-to-one correspondence between these two kinds of representations of unknown elements.} we have the following sequence of restricted chase derivations for this rule set, where $\theta$ is a substitution which maps $n_3$ to $n_1$ and other symbols to themselves.
From this derivation sequence, it can be seen that $I_3$ is not a new instance, and therefore, $(r_2, r_3, r_2)$ is not an active path, i.e., the one that leads to a (valid) restricted chase sequence.
	$$
	\begin{array}{ll}
	I_0 = \{\mathrm{HasKey}\big(a,b\big)\} \xlongrightarrow{\langle r_2,\{x/a, y/b\}\rangle}\\
	I_1 = I_0 \cup \big\{\mathrm{Enters}\big(a, n_1\big), \mathrm{KeyOpens}\big(b, n_1\big)\big\} \xlongrightarrow{\langle r_3, \{x/a, y/n_1\}\rangle}\\	
	I_2 = I_1 \cup \big\{\mathrm{HasKey}\big(a, n_2\big), \mathrm{KeyOpens}\big(n_2, n_1\big)\big\}
	\xlongrightarrow{\langle r_2, \{x/a, y/n_2\}\rangle}\\
	I_3 = I_2 \cup \big\{\mathrm{Enters}\big(a, n_3\big), \mathrm{KeyOpens}\big(n_2, n_3\big)\big\}
%{\color{red}\xlongrightarrow
	\xLongrightarrow{\theta=\{n_3/n_1\}}
	\theta(I_3) \subseteq I_2
%	\xlongrightarrow{\langle r_2, \{x/a, y/n_2\}\rangle}\\
%	I_3 = I_2 \cup \big\{Enters\big(a, n_3)\big), KeyOpens\big(n_2, n_3\big)\big\}\\
%	\xlongrightarrow{\langle r_3, \{x/a, y/n_3\}\rangle}%\\
%	{\color{red}I_4 = I_3 \cup \big\{HasKey\big(a, n_4\big), KeyOpens\big(n_4, n_3\big)\big\}}\\
%	{\color{red}\xlongrightarrow{\{n_4/n_2\}}}\\
%	{\color{red}I_4 = I_3}
	\end{array}
	$$
%}

From the above sequence of derivations, it can be seen that when we attempt to apply $r_2$ on $I_2$, its head can be instantiated to $\mathrm{Enters}(a,\_)$ and $\mathrm{KeyOpens}(n_2,\_)$, where we place an underline to mean that the existential variable $v$ in $r_3$ can be instantiated to form atoms that are already in $I_2$, which halts the derivation under the restricted chase.

In this paper, we will show that we can run such tests on cyclic rule applications of a fixed nesting depth, which we call $k$-cycles ($k > 0$), with the databases, which we call {\em restricted critical databases}, to define a hierarchy of classes of the finite restricted chase.

In addition, we show how to extend {\em $\delta$-bounded ontologies}, which were introduced in the context of the skolem chase variant \cite{zhang2015existential}, uniformly to $\delta$-bounded rule sets under the restricted chase variant, where $\delta$ is a bound function for the maximum depth of chase terms in a chase sequence.
Furthermore, as a concrete case of $\delta$, we consider functions constructed from an exponential tower of the length $\kappa$ (called $\mathsf{exp}_{\kappa}$ in this paper), for some given integer $\kappa$, and then we obtain the membership as well as reasoning complexities with these rule sets.

The main contributions of this paper are as follows:
\begin{enumerate}
	\item 
We show that while the traditional critical database technique \cite{marnette2009generalized}
	does not work for the restricted chase, a kind of ``critical databases" exist by which any finite restricted chase sequence can be faithfully simulated.
This is shown by Theorem \ref{key} for rules whose body contains no repeated variables (called {\em simple rules}) and by Theorem \ref{key0} for arbitrary rules.
	\item As the above results provide sufficient conditions to identify classes of the finite restricted chase, 
	we define a hierarchy of such classes, which can be instantiated to a concrete class of finite chase, given an acyclicity condition.
	This is achieved by Theorem \ref{thm:kSAFE} based on which various acyclicity conditions under the skolem chase can be generalized to introduce classes of finite chase beyond finite skolem chase.
	\item
	We show that the hierarchy of $\delta$-bounded rule sets under the skolem chase \cite{zhang2015existential}
	can be generalized by introducing $\delta$-bounded sets under the restricted chase.
	\item
	Our experimental results on a large set of ontologies collected from the web show practical applications of our approach to real-world ontologies. In particular, 
	in contrast with the current main focus of the field on acyclicity conditions for termination analysis, our experiments show that many ontologies in the real-world involve cycles of various kinds but indeed fall into the finite chase.
\end{enumerate}

The paper is organized as follows. The next section provides the preliminaries of the paper, including notations, some basic definitions, and a motivating example. 
Section \ref{previous} describes previous work on chase termination, which allows us to compare with the work of this paper during its development.  Then Section \ref{activeness} sets up the foundation of this work, namely on how to simulate restricted chase for any database by restricted chase with restricted critical databases.
We then define in Section \ref{k-safe} a hierarchy of classes of the finite restricted chase, called $k$-${\mathsf{safe}(\Phi)}$ rule sets for a given cycle function $\Phi$, by testing cycles of increasing nesting depths. In Section \ref{bounded} we apply a similar idea to $\delta$-bounded rule languages of 
\cite{zhang2015existential} and study membership checking and reasoning complexities. We implemented membership checking and a reasoning engine for $k$-${\mathsf{safe}(\Phi)}$ rule sets and conducted experiments. These are reported in Section \ref{experiments}. We then provide a further discussion on related work in Section \ref{8}. Finally, Section \ref{conclusion} concludes the paper with future directions.

This paper is a substantial revision and extension of a preliminary report of the work that appeared in \cite{karimi2018restricted}.

\section{Preliminaries}	
We assume the disjoint countably infinite sets of {\em constants} $\mathsf{C}$, 
({\em labelled}) {\em nulls} $\mathsf{N}$, 
{\em function symbols} $\mathsf{F}$, {\em variables} $\mathsf{V}$ and {\em predicates} $\mathsf{P}$. 
A {\em schema} is a finite set $\mathcal{R}$ of relation (or predicate) symbols. 
Each predicate or function symbol $Q$ is assigned a positive integer as its arity which is denoted by $arity(Q)$.
{\em Terms} are elements in $\mathsf{C}\cup\mathsf{N}\cup \mathsf{V}$.
An {\em atom} is an expression of the form $Q(\mathbf{t})$, where $\mathbf{t}\in (\mathsf{C}\cup \mathsf{V} \cup \mathsf{N})^{arity(Q)}$ and $Q$ is a predicate symbol from $\mathcal{R}$.
A {\em general instance} (or simply an {\em instance}) $I$ is a set of atoms over the schema ${\cal R}$; $term(I)$ denotes the set of terms occurring in $I$. 
A {\em database} is a finite instance $I$ where terms are constants from  $\mathsf{C}$.
A {\em substitution} is a function $h: \mathsf{C}\cup\mathsf{V}\cup\mathsf{N}\rightarrow \mathsf{C}\cup\mathsf{V}\cup\mathsf{N}$ such that (i) for all $c\in \mathsf{C}$, $h(c)=c$; (ii) for all $n\in \mathsf{N}$, $h(n)\in \mathsf{C}\cup \mathsf{N}$, and (iii) for all $v\in \mathsf{V}$, $h(v)\in \mathsf{C}\cup \mathsf{N}\cup \mathsf{V}$.
%{\color{blue}
Let $S_1$ and $S_2$ be sets of atoms over the same schema. A substitution $h: S_1\rightarrow S_2$ is called a {\em homomorphism} from $S_1$ to $S_2$ if $h(S_1)\subseteq S_2$ where $h$ naturally extends to atoms and sets of atoms. 
In this paper, when we define a homomorphism $h: S_1\rightarrow S_2$, if $S_1$ and $S_2$ are clear from the context, we may just define such a homomorphism as a mapping from terms to terms.
%}

A rule (also called a {\em tuple-generating dependency}) is a first-order sentence $r$ of the form: ${\forall {\mathbf{x}}\forall{\mathbf{y}}~ (\phi({\mathbf{x}},{\mathbf{y}})\rightarrow\exists{\mathbf{z}} ~\psi({{\mathbf{x},\mathbf{z}})})}$, where $\mathbf{x}$ and $\mathbf{y}$ are sets of universally quantified variables (in writing, we often omit the universal quantifier) and $\phi$ and $\psi$ are conjunctions of atoms constructed from relation symbols from ${\cal R}$, variables from $\mathbf{x} \cup \mathbf{y}$ and $\mathbf{x} \cup \mathbf{z}$, and constants from $\mathsf{C}$.
The formula $\phi$ (resp. $\psi$) is called the {\em body} of $r$, denoted $body(r)$ (resp. the {\em head} of $r$, denoted $head(r)$). 
In this paper, a rule set is a finite set of rules. These rules are also called {\em non-disjunctive rules} as compared to studies on disjunctive rules (see, e.g., \cite{bourhis2016guarded,carral2017restricted}).

%{\color{blue}
We implicitly assume all rules are {\em standardized apart} so that no variables are shared by more than one rule, even if, for convenience, we reuse variable names in examples of the paper.%}
A  rule is {\em simple} if variables do not repeat locally inside the body of the rule.
A {\em simple rule set} is a finite set of simple rules.

Given a rule $r = \phi(\mathbf{x},\mathbf{y})\rightarrow\exists\mathbf{z} \, \psi({\mathbf{x},\mathbf{z}})$, a {\em skolem function symbol} $f_z$ is introduced for each variable $z\in \mathbf{z}$, where $arity(f_z)=|\mathbf{x}|$. 
This leads to consider complex terms, called {\em skolem terms}, built from skolem functions and constants. However, in this paper, we will regard skolem terms as a special class of nulls (i.e., skolem terms will be seen as a way of naming nulls).

{\em Ground terms} in this context are constants from $\mathsf{C}$ or skolem terms, and atoms in a general instance may contain skolem terms as well. A {\em ground instance} in this context is a general instance involving no variables. 
The {\em functional transformation} of $r$, denoted $sk(r)$, is the formula obtained from $r$ by replacing each occurrence of $z\in \mathbf{z}$ with $f_{z}(\mathbf{x})$.
The {\em skolemized version} of a rule set $R$, denoted $sk(R)$, is the set of rules $sk(r)$ for all $r\in R$. 

Given a rule $r = \phi(\mathbf{x},\mathbf{y})\rightarrow\exists\mathbf{z} \, \psi({\mathbf{x},\mathbf{z}})$, we use 
$var_{u}(r)$, $var_{fr}(r)$, $var_{ex}(r)$, and $var(r)$, respectively, to refer to
the set of {\em universal} ($\mathbf{x}\cup \mathbf{y}$), {\em frontier} ($\mathbf{x}$), {\em existential} ($\mathbf{z}$), and {\em all} variables appearing in $r$. Given a rule set $R$, the {\em schema} of $R$ is denoted by $sch(R)$.
Given a ground instance $I$ and a rule $r$, an {\em extension} $h'$ of a homomorphism $h$ from $body(r)$ to $I$, denoted $h' \supseteq h$, is a homomorphism from $body(r)\cup head(r)$ to $I$, that assigns, in addition to the mapping $h$, ground terms to existential variables of $r$.
A {\em position} is an expression of the form $P[i]$, where $P$ is an $n$-ary predicate and $i$ ($1\le i\le n$) is an integer. We are interested only in positions associated with frontier variables \--  for each $x\in var_{fr}(r)$, $pos_{B}(x)$ (resp. $pos_{H}(x)$) denotes the set of positions of $body(r)$ (resp. $head(r)$) in which $x$ occurs.

We further define: a {\em path} $(r_1, r_2,\dots )$ (based on $R$) is a nonempty (finite or infinite) sequence of rules from  $R$; a {\em  cycle} $(r_{1}, \dots, r_{n})$ ($n\geq 2$) is a finite path whose first and last elements coincide  (i.e., $r_{1}=r_{n}$); a {\em  $k$-cycle} ($k \geq 1$)  is a cycle in which at least one rule has $k+1$ occurrences and all other rules have $k+1$ or less occurrences. 
Given a path $\pi$, $\mathsf{Rule}({\pi})$ denotes the set of distinct rules appearing in $\pi$.

%{\color{blue}
For a set or a sequence $W$, the cardinality $|W|$ is defined as usual. The size of an atom $p(\mathbf{x})$ is $|\mathbf{x}|$ and given a rule set $R$, with $||R||$, we denote the sum of the sizes of atoms in $R$.
%}

\subsection{Skolem and Restricted Chase Variants}
The chase procedure is a construction that accepts as input a database $I$ and a rule set $R$ and adds atoms to $I$.
In this paper, our main focus is on the skolem and the restricted chase variants.

We first define triggers, active triggers, and their applications.
The skolem chase is based on triggers, while the restricted chase applies only active triggers.

\begin{defn}
	Let $R$ be a rule set, $I$ an instance, and $r \in R$.
	A pair $(r,h)$ is called {\em a trigger for $R$ on $I$} (or simply {\em a trigger on $I$}, as $R$ is always clear from the context) if $h$ is a homomorphism from $body(r)$ to $I$. If in addition there is no extension $h' \supseteq h$, where  $h': body(r)\cup head(r) \rightarrow I$, then $(r,h)$ is called {\em an active trigger on $I$}.
	
	An {\em application} of a trigger $(r,h)$ on $I$ returns $I'=I \cup h(sk(head(r)))$. We write a trigger application by $I\langle r, h\rangle I'$, or alternatively by $I \xlongrightarrow{\langle r, h \rangle} I'$. We call atoms in $h(body(r))$ the {\em triggering atoms w.r.t. $r$ and $h$}, or simply {\em triggering atoms} when $r$ and $h$ are clear from the context.
\end{defn}

Intuitively, a trigger $(r,h)$ is active if given $h$, the implication in $r$ cannot be satisfied by any extension $h' \supseteq h$ that maps existentially quantified variables to terms in $I$. 

\begin{defn}
\label{Skolem-def}
	Given a database $I$ and a rule set $R$, we define the skolem chase based on a breadth-first fixpoint construction as follows: 
	we let {\small$\mathsf{chase}_{sk}^0(I,R)=I$} and, for all $i>0$, let {\small $\mathsf{chase}_{sk}^i(I,R)$} be the union of {\small $\mathsf{chase}_{sk}^{i-1}(I,R)$} and $h(head(sk(r)))$ for all rules $r \in R$ and all homomorphisms $h$ 
	such that $(r,h)$ is a trigger on {\small $\mathsf{chase}_{sk}^{i-1}(I,R)$}. Then, we let {\small $\mathsf{chase}_{sk}(I,R)$} be the union of {\small$\mathsf{chase}_{sk}^i(I,R)$}, for all $i \ge 0$.
\end{defn}

%{\color{blue}
Sometimes we need to refer to a {\em skolem chase sequence}, which is a sequence of instances that starts from a database $I_0$ and continues by applying triggers for the rules in a given path on the instance constructed so far.
The term {\em skolem chase sequence} in this case is independent of whether such a sequence can be extended to an infinite sequence or not.%}
%and keep doing this until a fixpoint is reached.
We can also distinguish the two cases where the chase is terminating or not.

A finite sequence of rule applications from a path $(r_1, \dots, r_n)$ produces a finite sequence of instances $I_0, I_1, \dots, I_n$ such that (i) $I_{i-1}\langle r_i, h_i\rangle I_i$, where $(r_i, h_i)$ is a trigger on $I_{i-1}$ for all $1 \le i \le n$, (ii) there is no trigger $(r, h)$ on $I_n$ such that $(r, h)\notin \{(r_i, h_i)\}_{0\le i\le n-1}$,
%has already been applied on $I_i$ for $i<n$, 
%every trigger on $I_n$ has already been applied on $I_i$ for $i<n$, 
and (iii) for each $1 \le i < j \le n$, assuming that $I_{i-1} \langle r_i, h_i\rangle I_i$ and $I_{j-1} \langle r_j, h_j\rangle I_j$, $r_i=r_j$ implies $h_i\neq h_j$ (i.e., homomorphism $h_i$ is different from $h_j$).
The result of the chase sequence is $I_n$.

An infinite sequence $I_0, I_1, \dots$ of instances is said to be a {\em non-terminating skolem chase sequence} if (i) for all $i\ge 1$, there exists a trigger $(r_i,h_i)$ on $I_{i-1}$ such that $I_{i-1}\langle r_i, h_i\rangle I_i$, (ii) for each $i,j\ge 1$ such that $i\ne j$, assuming that $I_{i-1} \langle r_i, h_i\rangle I_i$ and $I_{j-1} \langle r_j, h_j\rangle I_j$, $r_i=r_j$ implies $h_i\neq h_j$.\footnote{In the literature, in addition to (i) and (ii), another condition known as the {\em fairness condition for the skolem chase} is imposed: for each $i\ge 0$, and each trigger $(r_i,h_i)$ on $I_{i-1}$, there exists some $j\ge i$ such that $I_{j-1}\langle r_i, h_i\rangle I_j$. This last condition guarantees that all the triggers are eventually applied. We remove this requirement, as for the case of the skolem chase, this condition is immaterial, cf. \cite{gogacz2019all}.}
In this case, the result of the chase sequence is $\bigcup_{i \geq 0} I_i$.

From \cite{marnette2009generalized}, we know that if some skolem chase sequence of a rule set $R$ and a database $I_0$ terminates, then all instances returned by {\em any} skolem chase sequence of $I_0$ and $R$ are terminating, and are the same.

On the other hand, the restricted chase is known to be order-sensitive. For this reason, it is defined only on sequences of rule applications.

Similar to a skolem chase sequence, the main idea of a restricted chase sequence (based on a given rule set $R$) is starting from a given database and applying triggers for the rules in a path based on $R$ on the instance constructed so far.
However, unlike the skolem chase sequence, only active triggers are applied.
Similar to the case of the skolem chase, we distinguish the two cases where the chase is terminating or not.

\begin{defn} \label{restricted}
Let $R$ be a rule set and $I_0$ a database.
%Similar to a skolem chase sequence, the main idea of a restricted chase sequence (based on $R$) is starting from $I_0$ and applying triggers for $R$ on the instance constructed so far until we reach a fixpoint. However, unlike the skolem chase sequence, only active triggers are applied.
%Similar to the case of the skolem chase, we distinguish the two cases where the chase is terminating or not:
%A {\em restricted chase sequence} (based on $R$) is a sequence of rule applications from $R$ which produces a sequence of instances $I_0, I_1, \dots, I_n$ such that $I_{i-1}\langle r_i, h_i\rangle I_i$, where $(r_i,h_i)$ is an active trigger on $I_{i-1}$, for all $1\le i\le n$. Furthermore, we have
	\begin{itemize}
	\item
	A finite sequence $I_0, I_1,\dots, I_n$ of instances is called {\em a terminating restricted chase sequence (based on $R$)} if (i) for each $1\le i\le n$ there exists an active trigger $(r_i,h_i)$ on $I_{i-1}$ such that $I_{i-1}\langle r_i, h_i\rangle I_i$; and (ii) there exists no active trigger on $I_n$.
	%; and (iii) for each $1 \le i < j \le n$, assuming that $I_{i-1} \langle r_i, h_i\rangle I_i$ and $I_{j-1}\langle r_j, h_j\rangle I_j$, $r_i=r_j$ implies $h_i\neq h_j$.
	The result of the chase sequence is $I_n$.
		\item
		An infinite sequence $I_0,I_1,\dots$ is called a {\em non-terminating (or infinite) restricted chase sequence (based on $R$)}
		if
		\begin{itemize}
			\item[(i)] for each $i \geq 0$, there exists an active trigger $(r_i,h_i)$ on $I_{i-1}$ such that $I_{i-1}\langle r_i, h_i\rangle I_i$;
			%\item[(ii)] for each $1 \le i < j \le n$, assuming that $I_{i-1}\langle r_i, h_i\rangle I_i$ and $I_{j-1}\langle r_j, h_j\rangle I_j$, $r_i=r_j$ implies $h_i\neq h_j$;
			and
			\item[(ii)] it satisfies the {\em fairness condition}: for all $i \geq 1$ and all active triggers $(r_i,h_i)$ on $I_{i-1}$, where $r_i \in R$, there exists $j \geq i$ such that either $I_{j-1} \langle r_i,h_i\rangle I_j$ or the trigger $(r_i,h_i)$ is not active on $I_{j-1}$.
		\end{itemize}
		The result of the chase sequence is $\bigcup_{i \geq 0} I_i$.
	\end{itemize}
\end{defn}
%}

\begin{exm}
	Let us consider instance $I=\{P(a,b), P(b,c), P(c,a), Q(a,b)\}$ and rule $r$:
	$$\begin{array}{ll}
	r: P(x,y), P(y,z), P(z,x)\rightarrow \exists u\, Q(x,u)
	\end{array}$$
	Homomorphism $h_1=\{x/a, y/b, z/c\}$ maps $body(r)$ to $I$. The pair $(r, h_1)$ is a trigger on $I$ and $I\langle r, h_1\rangle I\cup \{Q(a,f_u(a))\}$ where $f_u$ is a skolem function constructed from $u$. However, $(r, h_1)$ is not active for $I$.	
	On the other hand, homomorphism $h_2=\{x/c, y/a, z/b\}$ maps $body(r)$ to $I$ and $(r, h_2)$ is an active trigger on $I$ since there is no extension $h'_2$ of $h_2$ such that $h'_2(head(r))\subseteq I$. So, we have $I\langle r, h_2\rangle I\cup \{Q(c,f_u(c))\}$.
	Therefore, $(r, h_1)$ can be applied for the skolem chase but not for the restricted chase while 
$(r, h_2)$	 can be applied for both chase variants.
\end{exm}

Note that the fairness condition essentially says that any active trigger is eventually either applied or becoming inactive. Furthermore, an infinite restricted chase sequence $\mathcal{I}$ cannot be called non-terminating if the fairness condition is not satisfied for $\mathcal{I}$.
Recently, in \cite{gogacz2019all}, it has been shown that for rules with {\em single heads} (i.e., where the head of a rule consists of a single atom), the fairness condition can be safely neglected.

A rule set $R$ is said to be {\em (all-instance) terminating} under the restricted chase, or simply {\em restricted chase terminating} if it has no infinite restricted chase sequence w.r.t.\,all databases; otherwise, $R$ is non-terminating under the restricted chase; this is the case where there exists at least one non-terminating restricted chase sequence w.r.t.\,some database. 
	
The classes of rule sets whose chase terminates on all paths (all possible derivation sequences of chase steps) independent of the given databases (thus all instances) are denoted by $\mathsf{CT}_{\forall\forall}^{\triangle}$, where $\triangle\in\{\mathsf{sk},\mathsf{res}\}$ ($\mathsf{sk}$ for the skolem chase and $\mathsf{res}$ for the restricted chase).

Since a chase sequence is generated by a sequence of rule applications, sometimes it is convenient to talk about a chase sequence in terms of a sequence of rules that are applied. On the other hand, to each path can be assigned a chase sequence (which is not unique).

For convenience, given a finite path $\pi = (r_1, \dots, r_n)$ based on $R$ and database $I_0$, we say that {\em $\pi$ leads to a weakly restricted chase sequence (of $R$ and $I_0$)} if there are active triggers $(r_i, h_i)$ on $I_{i-1}$ $(1 \le i \le n)$ such that $I_{i-1}\langle r_i, h_i\rangle I_i$. Note that the condition is independent of whether there exists an active trigger on $I_{n}$ or not; so, we do not qualify the sequence $I_0, I_1, \dots, I_n$ as terminating or non-terminating.  Furthermore, the condition only requires the existence of active triggers and does not mention whether the fairness condition is satisfied or not in case the sequence can be expanded to an infinite one. By abuse of terminology, we will drop the word {\em weakly} in the rest of this paper when no confusion arises; this is not a technical concern related to deciding on the finite restricted chase in our approach since our approach is based on certain types of terminating restricted chase sequences.

Finally, a {\em conjunctive query} (CQ) $q$ is a formula of the form $q(\mathbf{x}) := \exists \mathbf{y} \, \Phi(\mathbf{x}, \mathbf{y})$, where $\mathbf{x}$ and $\mathbf{y}$ are tuples of variables and $\Phi(\mathbf{x},\mathbf{y})$ is a conjunction of atoms with variables in $\mathbf{x}\cup \mathbf{y}$. A {\em Boolean conjunctive query} (BCQ) is a CQ of the form $q()$.
It is well-known that, for all BCQs $q$ and for all databases $I$, $I\cup R\models q$ (under the semantics of first-order logic) if and only if $q$ is entailed by the result of the chase on $R$ and $I$ for either the semi-oblivious or the restricted chase variant \cite{fagin2005data}.

\subsection{A Concrete Example}
To illustrate the practical relevance of the restricted chase and also use it as a running example, let us consider modeling a secure communication protocol where two different signal types can be transmitted: type $A$ for inter-zone communication and type $B$ for intra-zone communication. Let us consider a scenario where a transmitter from one zone requests to establish secure communication with a receiver from another zone in this network.
There is an unknown number of trusted servers. Before a successful communication between two users can occur, following a handshake protocol, the transmitter must send a type $A$ signal to a trusted server in the same zone and receive an acknowledgment back. Then, that trusted server sends a type $B$ signal to a trusted server in the receiver zone.

Figure \ref{fig:Msg_Scenario} illustrates the above data transmission scenario where there are just two cells in each of which there are several users (solid dark circles) and base stations (under blue boxes). If a transmitter $\mathsf{t}$ in cell 1 requests to transmit a data message to a receiver $\mathsf{r}$ in cell 2, then $\mathsf{t}$ must establish a handshake protocol to some base station (e.g., $\mathsf{b1}$) in the same cell (sending and receiving to/from $\mathsf{b1}$). After a handshake protocol is established, $\mathsf{b1}$ sends a data message to some base station in cell 2 ($\mathsf{b2}$ in the figure) to complete the required communication before $\mathsf{t}$ sends a data message to $\mathsf{r}$.

	\begin{figure}[hbt!]
		\vspace{-25pt}
		\begin{center}
			\includegraphics[width=280pt]{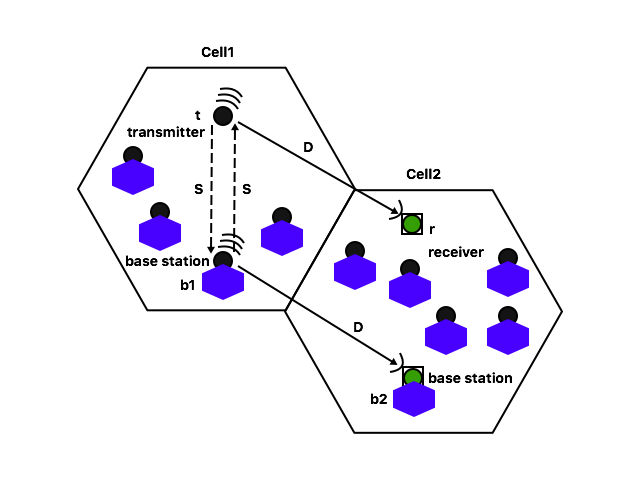}
		\end{center}
		\vspace{-30pt}
		\caption[Data transmission scenario]
		{Data transmission scenario.\label{fig:Msg_Scenario}}
	\end{figure}

Below, we use existential rules to model the required communication protocol (the modeling here does not include the actual process of transmitting signals).
Let us assume by default that every server is trusted. 

\begin{exm} \label{exm:simplified}
	Consider the rule set $R_1=\{r_1,r_2\}$ below
	 and its skolemization, where 
	$\mathrm{TypeA}(x,y)$ denotes a request to send a type $A$ signal from $x$ to $y$ and $\mathrm{TypeB}(x,y)$ a request to send a type $B$ signal from $x$ to $y$.
	$$\begin{array}{ll}
	r_1:  \mathrm{TypeB}(x,y) \rightarrow \exists u\, \mathrm{TypeA}(x,u), \mathrm{TypeA}(u,x)\\
	r_2: \mathrm{TypeB}(x,y), \mathrm{TypeA}(x,z), \mathrm{TypeA}(z,x) \rightarrow \exists v\,\mathrm{TypeB}(z,v)\\
	\\
	sk(r_1):  \mathrm{TypeB}(x,y) \rightarrow \mathrm{TypeA}(x,f_u(x)), \mathrm{TypeA}(f_u(x),x)\\
	sk(r_2): \mathrm{TypeB}(x,y), \mathrm{TypeA}(x,z), \mathrm{TypeA}(z,x) \rightarrow \mathrm{TypeB}(z,f_v(z))\\
	\end{array}
	\vspace{.1in}
	$$
where $f_u$ and $f_v$ are skolem functions constructed from $u$ and $v$, respectively.
	
	With database $I_0 = \{\mathrm{TypeB}(t,r)\}$, after applying $sk(r_1)$ and $sk(r_2)$ under the restricted chase, we get:
	$$
	\begin{array}{ll}
	I_0 = \{\mathrm{TypeB}(t,r)\} \xlongrightarrow{\langle sk(r_1),\{x/t, y / r\}\rangle}\\
	I_1 = I_0 \cup \{\mathrm{TypeA}(t, f_u(t)), \mathrm{TypeA}(f_u(t),t)\} \xlongrightarrow{\langle sk(r_2), \{x/t, y/r, z/ f_u(t)\}\rangle}\\
	I_2 = I_1 \cup \{\mathrm{TypeB}(f_u(t), f_v(f_u(t)))\}
	\end{array}
	\vspace{.1in}
	$$
That is, path $\pi_1 = (sk(r_1), sk(r_2))$ leads to a restricted chase sequence. But this is not the case for the path $\pi_2 = (sk(r_1), sk(r_2), sk(r_1))$, since the trigger for applying the last rule on the path is not active \-- with $ \mathrm{TypeB}(f_u(t),f_v(f_u(t)))$ as the triggering atom for the body of rule $sk(r_1)$,
its head can be satisfied by already derived atoms in $I_2$, namely, $\mathrm{TypeA}(f_u(t), t)$ and $\mathrm{TypeA}(t, f_u(t))$ (i.e., the existential variable $u$ in $sk(r_1)$ can be instantiated to $t$ so that the rule head is satisfied by $I_2$). 

To illustrate more subtleties, let us consider a slightly enriched rule set $R_2=\{r_3,r_4\}$. 
The difference from $R_1$ is that here we use a predicate $\mathrm{TrustedServer}(a)$ to explicitly specify that $a$ is a trusted server.
	$$\begin{array}{ll}
	r_3:  \mathrm{TypeB}(x,y) \rightarrow \exists u\, \mathrm{TrustedServer}(u), \mathrm{TypeA}(x,u), \mathrm{TypeA}(u,x)\\
	r_4: \mathrm{TypeB}(x,y), \mathrm{TypeA}(x,z), \mathrm{TypeA}(z,x) \rightarrow \exists v\, \mathrm{TrustedServer}(v), \mathrm{TypeB}(z,v)\\
	\\
	sk(r_3):  \mathrm{TypeB}(x,y) \rightarrow \mathrm{TrustedServer}(f_u(x)), 
	\mathrm{TypeA}(x,f_u(x)), \mathrm{TypeA}(f_u(x),x)\\
	sk(r_4): \mathrm{TypeB}(x,y), \mathrm{TypeA}(x,z), \mathrm{TypeA}(z,x) \rightarrow \mathrm{TrustedServer}(f_v(z)), \mathrm{TypeB}(z,f_v(z))
	\end{array}
	$$
	
	With the same input database $I_0$, 
	we can verify that any non-empty prefix of the 2-cycle
	$\sigma = (sk(r_3), sk(r_4), sk(r_3), sk(r_4), sk(r_3))$ leads to a restricted chase sequence except $\sigma$ itself. 
	Let us provide some details. 
	$$
	\begin{array}{ll}
	I_0 = \{\mathrm{TypeB}(t,r)\} \xlongrightarrow{\langle sk(r_3),\{x/t, y / r\}\rangle}\\
	I_1 = I_0 \cup \{\mathrm{TypeA}(t, f_u(t)), \mathrm{TypeA}(f_u(t),t)\} \xlongrightarrow{\langle sk(r_4), \{x/t, y/r, z/ f_u(t)\}\rangle}\\
	I_2 = I_1 \cup \{\mathrm{TypeB}(f_u(t), f_v(f_u(t)))\}
	\end{array}
	\vspace{.1in}
	$$	
	Observe that at this stage, since $t$ is not known as a trusted server (i.e., we do not have $\mathrm{TrustedServer}(t)$ in the given database), unlike the case of $R_1$, we are not able to instantiate the existential variable $u$ to $t$ to have the rule head satisfied. Thus, the restricted chase continues:
	$$
	\begin{array}{ll}
	I_3 = I_2 \cup \{\mathrm{TrustedServer}(f_u^2(t)), \mathrm{TypeA}(f_u(t), f_u^2(t)), \mathrm{TypeA}(f_u^2(t), f_u(t))\} \xlongrightarrow{\langle sk(r_4), \{x/f_u(t), y/f_v(f_u(t)), z/f_u^2(t)\}\rangle} \\
	I_4 = I_3 \cup \{\mathrm{TrustedServer}(f_v(f_u^2(t)), \mathrm{TypeB}(f_u^2(t)), f_v(f_u^2(t)))\}
	\end{array}
	$$
	Now, the pair $(sk(r_3), \{x/f_u^2(t), y/f_v(f_u^2(t))\})$ is a trigger on $I_4$. However, since the existential variable $u$ in $r_3$ can be instantiated to the skolem term $f_u(t)$ so that the head of $r_3$ is satisfied, the trigger is not active on $I_4$ and thus the chase terminates.
\end{exm}
%}

Figures \ref{fig:SkChaseOnR2} and \ref{fig:ResChaseOnR2} illustrate the skolem and the restricted chase on the rule set $R_2$, where an arrow over a relation symbol indicates a newly derived atom, or an existing atom used to satisfy the rule head so that the restricted chase terminates.  In contrast, while $R_2$ is non-terminating under the skolem chase, it can be shown that it is all-instance terminating under the restricted chase. 

\begin{figure}[hbt!]
\label{fig0}
		\vspace{3mm}
		\centering
		\hspace{-2.5cm}
		\begin{minipage}[t]{4cm}
			\centering
			\includegraphics[scale=0.38]{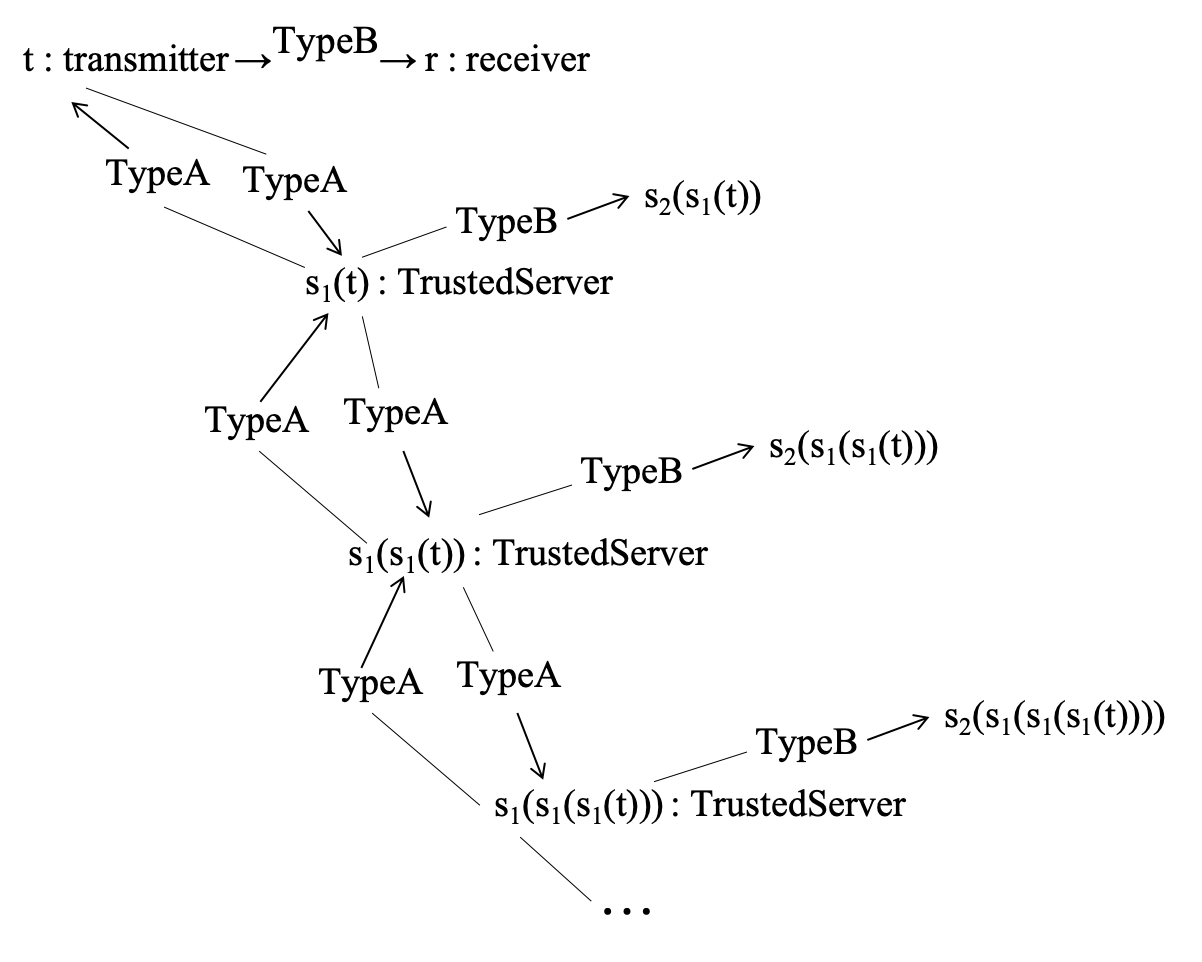}
			\caption{\hspace{-0.25mm}\small{Skolem chase on $R_2$}}
			\label{fig:SkChaseOnR2}
		\end{minipage}
		\hspace{2cm}
		\begin{minipage}[t]{4cm}
			\centering
			\vspace{-6.2cm}
			\includegraphics[scale=0.35]{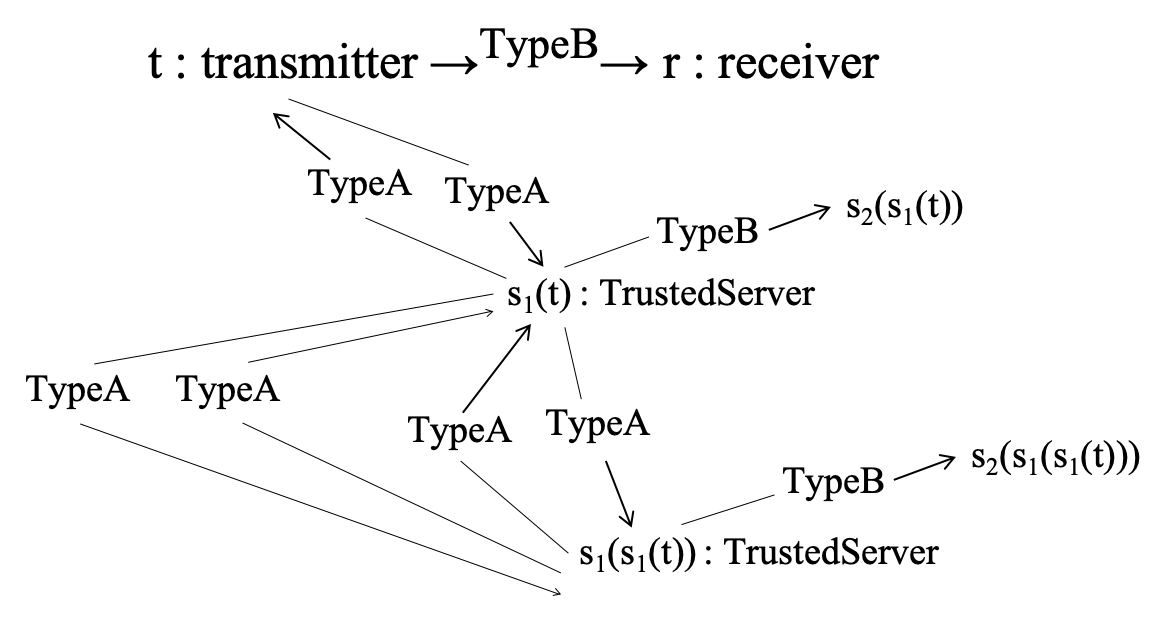}
			\vspace{2cm}
			\caption{\hspace{-0.25mm}\small{Restricted chase on $R_2$}}
			\label{fig:ResChaseOnR2}
		\end{minipage}
	\end{figure}

\section{Previous Development and Related Work}
\label{previous}

Since our technical development is often related to, or compared with, the state-of-the-art, let us introduce some key classes of the finite chase here and comment on the latest developments related to our work. 
Note that all acyclicity conditions that are given below ensure the termination of the skolem chase, and therefore, of the restricted chase, except for RMFA and RJA which ensure the termination of the restricted chase and allow to identify more terminating rule sets.

{\em Weakly-acyclic} (WA) \cite{fagin2005data}, roughly speaking, tracks the propagation of terms in different positions. A rule set is WA if there is no position in which skolem terms including skolem functions can be propagated cyclically, possibly through other positions.

{\em Joint-acyclic} (JA) \cite{krotzsch2011extending} generalizes WA as follows. Let $R$ be a rule set. For each variable $y\in var_{ex}(R)$, let $Move(y)$ be the smallest set of positions such that (i) $pos_H(y)\subseteq Move(y)$; and (ii) for each rule $r\in R$ that $var_{ex}(r)\neq \emptyset$ and for all variables $x\in var_u(r)$, if $pos_B(x)\subseteq Move(y)$, then $pos_H(x)\subseteq Move(y)$. The {\em JA dependency graph} $\JA(R)$ of $R$ is defined as: the set of vertices of $\JA(R)$ is $var_{ex}(R)$, and there is an edge from $y_1$ to $y_2$ whenever the rule that contains $y_2$ also contains a variable $x\in var_u(R)$ such that $pos_H(x)\neq \emptyset$ and $pos_B(x)\subseteq Move(y_1)$. $R\in \JA$ if $\JA(R)$ does not have a cycle.

A rule set $R$ belongs to the {\em acyclic graph of rule dependencies} (aGRD) class of acyclic rules if there is no cyclic dependency relation between any two (not necessarily different) rules of $R$, possibly through other dependent rules of $R$.
To define the rule dependency graph \cite{baget2004improving,baget2011rules} of a rule set $R$, we introduce the rule dependency relation $\prec\!\ \subseteq R\times R$ as follows. Consider two rules $r_1, r_2\in R$ such that $r_1=body(r_1)\rightarrow \exists \mathbf{z_1}\!\ head(r_1)$ and  $r_2=body(r_2)\rightarrow \exists \mathbf{z_2}\!\ head(r_2)$. Let $sk(r_1)= body(r_1)\rightarrow sk(head(r_1))$ and $sk(r_2)= body(r_2)\rightarrow sk(head(r_2))$. Then, $r_1\prec r_2$ if and only if there exists an instance $I$, substitutions $\theta_1$ (resp. $\theta_2$), for all variables in $sk(r_1)$ (resp. $sk(r_2)$) such that $\theta_1(body(r_1))\subseteq I$, $\theta_2(body(r_2))\subseteq I\cup \theta_1(sk(head(r_1)))$, and $\theta_2(body(r_2))\nsubseteq I$. $R$ has an acyclic graph of rule dependencies if $\prec$ on $R$ is acyclic. In this case, $R$ is called aGRD.

Note that the original definition of aGRD in \cite{baget2004improving} considers fresh nulls as opposed to skolem terms, which based on \cite{grau2013acyclicity} does not change the resulting relation $\prec$.\footnote{We will have more remarks on rule dependency and the important role it plays in our approach, 
after Definition \ref{depends-on}.}

{\em Model-faithful acyclic} (MFA) \cite{cuenca2013acyclicity} is a semantic acyclicity class of the skolem chase which generalizes all the skolem acyclicity classes mentioned above.
A rule set $R$ is MFA if in the skolem chase of $R$ w.r.t. the critical database of $R$ (i.e., the database which contains all possible ground atoms based on predicates of $R$ and the single constant symbol $\ast$ without any occurrence in $R$), there is no cyclic skolem term (a term with at least two occurrences of some skolem function). 

Also, {\em restricted model-faithful acyclicity} (RMFA) \cite{carral2017restricted} generalizes MFA as follows.
Let $R$ be a non-disjunctive rule set. For each rule $r\in R$ and each homomorphism $h$ such that $h$ is a homomorphism on $body(r)$, $C_{h,r}$ is defined as the union of $h(body(r))$, where each occurrence of a constant is renamed so that no constant occurs more than once, and $F_t$ for each skolem term $t$ in $h(body(r))$, where $F_t$ is the set of ground atoms involved in the derivation of atoms containing $t$.
Let RMFA($R$) be the least set of ground atoms such that it contains the critical database of $R$ and let $r\in R$ be a rule and $h$ a homomorphism from $body(r)$ to $\text{RMFA}(R)$. Let $v\in var_{ex}(r)$ be some existential variable of $r$. If $\exists v. h(head(r))$ is not logically entailed by the exhaustive application of non-generating (Datalog) rules on the set of atoms $C_{h,r}$, then $h(sk(head(r)))\subseteq \text{RMFA}(R)$. We define $R \! \in \! \text{RMFA}$ if $\text{RMFA}(R)$ contains no cyclic skolem terms.

In \cite{carral2017restricted}, a notion known as {\em restricted model-faithful acyclicity} (RMFC) has been introduced which provides a sufficient condition for deciding non-termination of the restricted chase of a given rule set for all databases. Intuitively, RMFC is based on detecting cyclic functional terms in the result of the exhaustive application of {\em unblockable} rules on the grounded version of $body(r)\cup sk(head(r))$ for some generating rule $r$, such that in the mapping used for the grounding, each variable $x$ is replaced by some fresh constant $c_x$.

To characterize a sufficient condition of termination of a given rule set for arbitrary databases, for any chase variant, it would be useful to have a special database that can serve as a {\em witness} for proving termination. Let us refer to it as a {\em critical database} $I^{\ast}$. Having such a critical database in place guarantees that given a rule set $R$, if there is some database that witnesses the existence of an infinite chase derivation of $R$, then $I^{\ast}$ is already such a witness database. If we know that such a critical database exists for some chase variant, then we can focus on sufficient conditions to decide the chase termination of a rule set w.r.t. $I^{\ast}$. 

From \cite{marnette2009generalized}, it is known that such a critical database exists for the oblivious and skolem chase variants. The construction of such a critical database for those chase variants is also easy:
Let $R$ be a rule set. Let $C$ denote the set of constants appearing in $R$ and let $\ast$ be a special constant with no occurrence in $R$. A database is a (skolem) critical database if each relation in it is a full relation on the domain $C\cup \{\ast\}$.
With this measure in place, it is then easy to show why all the known classes of terminating rule sets under skolem and oblivious chase variants (such as the aforementioned acyclicity conditions) work well. The reason is that they rely on this critical database.

However, for the restricted chase, no critical database exists. Note that for the terminating conditions of RJA and RMFA \cite{carral2017restricted} that are the only known concrete criteria for the termination of restricted chase rules, the introduced ``critical databases" are ad-hoc in a way that they do not provide a {\em principled way} to construct such a database that may lead to more general classes of terminating rule sets under the restricted chase. 
In fact, due to the nature of the problem, which is not recursively-enumerable \cite{DBLP:journals/fuin/GrahneO18}, as also pointed out in \cite{gogacz2019all}, finding such a critical database even for subsets of rules with syntactic (or semantic) restrictions is very challenging.
More recently, termination of linear rules under the restricted chase has been considered in \cite{leclere2019single}, where the body and the head of rules are composed of singleton atoms (called {\em single-body} and {\em single-head} rules). As part of this work, the existence of such a critical database is proved by simply showing a database consisting of a single atom.

Also, for single-head {\em guarded} and {\em sticky} rules, the same problem has been considered in \cite{gogacz2019all}, where the authors characterize non-termination of restricted chase sequences constructed from the aforementioned rule sets using sophisticated objects known as {\em chaseable sets} which are infinite in size. For this purpose, they show that the existence of an infinite chaseable set characterizes the existence of an infinite restricted chase derivation. In particular, for guarded rule sets, the latter can be strengthened with the fact that we can focus on acyclic databases to show the decidability of restricted chase termination for guarded TGDs. 

Furthermore, for sticky TGDs, this is shown via the existence of a {\em finitary caterpillar} which is an infinite {\em path-like} restricted chase derivation of some database the existence of which can be checked via a deterministic {\em Büchi automaton}. 
Their work is focused only on single-head rules and, to the best of our knowledge, no characterization exists for multi-head rules.
This is unlike the skolem chase for which the notion of $\delta$-bounded ontologies have been defined uniformly using the (skolem) critical database technique \cite{zhang2015existential}. 

The decision problem of termination of the oblivious and the skolem chase variants have been considered for linear and guarded rules in \cite{calautti2015chase}, and this problem is shown to be \textsc{PSpace}-complete and 2\textsc{ExpTime}-complete, respectively, for linear and guarded rules.
More recently, the same problem has been considered for sticky rules in \cite{calautti2019oblivious}, and it has been shown to be \textsc{PSpace}-complete. This shows that for these rules, sufficient and necessary conditions can be established to decide termination.

It is worthy to mention that similar to our work, in \cite{baget2014extending}, a tool was introduced to extend different (skolem) acyclicity conditions ensuring chase termination. However, unlike our approach, their extension never extends a skolem chase terminating rule set to a terminating one under the restricted chase. Their extension is also without increasing the complexity upper bound of the membership checking problem of the original rules.

Also related to this work, the notion of {\em $k$-bounded rules} was introduced in \cite{delivorias2018k} for oblivious, skolem and restricted chase variants. The $k$-boundedness problem they considered in that work checks whether, independently from any database, there is a fixed upper bound of size $k$ on the number of breadth-first chase steps for a given rule set, where $k$ is an integer.  For arbitrary values of $k$, this problem is already known to be undecidable for Datalog rules (TGDs without existential variables, also known as range-restricted TGD \cite{abiteboul1995foundations}), as established in \cite{hillebrand1995undecidable,marcinkowski1999achilles}.

The breadth-first chase procedure in \cite{delivorias2018k} refers to chase sequences in which rule applications are prioritized. Their prioritization is in a way that those rule applications which correspond to a particular breadth-first level occur before those that correspond to a higher breadth-first level.
Under the assumption that $k$ is excluded from the input, and only the rule set is given as the input, they prove an \textsc{ExpTime} upper bound for checking $k$-boundedness for the oblivious and the skolem chase variants and \textsc{2ExpTime} upper bound for the restricted chase.\footnote{Note, however, that if $k$ is part of the input, i.e., when the problem is: {\em given a rule set $R$ and a unary-encoded integer $k$, whether $R$ is $k$-bounded for the considered chase}, the complexity of the problem is in \textsc{2ExpTime} and \textsc{3ExpTime} for the aforementioned chase variants, respectively.}

Notice that as discussed in \cite{delivorias2018k}, TGDs with $k$-boundedness property are {\em union of conjunctive queries-rewritable} (or {\em UCQ-rewritable}, also known to belong to {\em finite unification sets} of TGDs (or {\em fus}) \cite{baget2011rules}).
It is worth mentioning that this latter work has a different scope from ours in that, unlike $k$-bounded TGDs of \cite{delivorias2018k}, the $k$-$\mathsf{safe}(\Phi_{\Delta})$ rule sets that result from the current paper, where $\Delta$ is some skolem acyclicity condition, already generalize Datalog rule sets (for any value of $k\ge 0$), and therefore, are not UCQ-rewritable. Besides, there is no characterization of any critical database for the restricted chase variant in \cite{delivorias2018k} which is a key issue and the focus of the current paper.

\section{Finite Restricted Chase by Activeness} \label{activeness}
In this section, we tackle the question of what kinds of tests we can do to provide sufficient conditions to identify classes of the finite restricted chase. 
With this goal in mind, we present the notion of the restricted critical database for a given path and show that any ``chained" restricted chase sequence for a given path w.r.t. an arbitrary database can be simulated by using the restricted critical database for simple rules and by using an {\em updated restricted critical database} via renaming for arbitrary rules. 

\subsection{Restricted Critical Databases and Chained Property}
A primary tool for termination analysis of the skolem chase is the technique of critical database \cite{marnette2009generalized}.
Recall that, given a rule set $R$, if $C$ denotes the set of constants which occur in $R$, the {\em critical database} (or {\em skolem critical database}) of $R$, denoted $I^R$, is a database defined in a way that each relation in $I^R$ is a full relation on the domain $C\cup \{\ast\}$, in which $\ast$ is a special constant with no occurrence in $R$.
The critical database can be used to faithfully simulate termination behavior of the skolem chase \-- a rule set is all-instance terminating if and only if it is terminating w.r.t. the skolem critical database. 
However, this technique does not apply to the restricted chase.

\begin{exm} \label{exm:I_P}
	Given a rule set $R = \{E(x_1,x_2)\rightarrow\exists z\,  E(x_2,z)\}$ and its critical database $I^{R}=\{E({\ast},{\ast})\}$, where ${\ast}$ is a fresh constant, the skolem chase does not terminate w.r.t.\,$I^R$, which is a faithful simulation of the termination behavior of $R$ under the skolem chase. But the restricted chase of $R$ and $I^{R}$ terminates in zero step, as no active triggers exist. However, the restricted chase of $R$ and database $\{E(a,b)\}$ does not terminate.
\end{exm}

	The above example is not at all a surprise, as the complexity of membership checking in the class of rule sets that have a finite restricted chase, namely whether a rule set is in $\mathsf{CT}_{\forall\forall}^\mathsf{res}$, is $\mathsf{coRE}$-hard \cite{DBLP:journals/fuin/GrahneO18}, which implies that in general there exists no effectively computable (finite) set of databases which can be used to simulate termination behavior w.r.t.\,all input databases, as otherwise the membership checking for $\mathsf{CT}_{\forall\forall}^\mathsf{res}$ would be recursively enumerable, a contradiction to the $\mathsf{coRE}$-hardness result of 
	\cite{DBLP:journals/fuin/GrahneO18}.
			
	To check for termination, one natural consideration is the notion of cycles based on a given rule set. Firstly, a chase that terminates w.r.t. a database $I$ on all $k$-cycles implies chase terminating w.r.t. $I$ on all $k'$-cycles, for all $k' >k$. This is because a chase that goes through a $k'$-cycle must go through at least one $k$-cycle. Secondly, since a non-terminating chase must apply at least one rule infinitely many times, if the termination is guaranteed for all $k$-cycles for a fixed $k$, then an infinite chase becomes impossible. Thus, testing all $k$-cycles can serve as a means to decide classes of the finite chase.
Furthermore, cycles are recursively enumerable with increasing lengths and levels of nesting.  
We can test $(k+1)$-cycles for a possible decision of the finite restricted chase when such a test failed for $k$-cycles. We, therefore, may find larger classes of terminating rule sets with an increased computational cost. We have demonstrated this approach in Example \ref{exm:simplified}, where the rule set $R_2$ is terminating on all $2$-cycles but not on some $1$-cycles. 
However, a challenging question is {\em which databases to check against}. In the following, we tackle this question.

Given a path, our goal is to simulate a sequence of restricted chase steps with an arbitrary database by a sequence of restricted chase steps with a fixed database. 
On the other hand, since in general we can only expect sufficient conditions for termination, such a simulation should at least capture all infinite derivations by a rule set with an arbitrary database.
On the other hand, we only need to consider the type of paths that potentially lead to cyclic applications of the chase. In the following, we will address this question first. 

\begin{exm} \label{exm:counter}
	Consider the singleton rule set $R$ with rule
	$
	r: 	~T(x,y),P(x,y)\rightarrow \exists z \, T(y,z)$ and its skolemization $sk(r): T(x,y),P(x,y)\rightarrow T(y,f_z(y)).$
With $I_0=\{T(a,b),P(a,b)\}$, we have:
	{\small $\mathsf{chase}_{sk}(I_0,R)$} $=I_0\cup \{T(b,f_z(b))\}$.
	After one application of $r$, no more triggers exist and thus the skolem chase of $R$ and $I_0$ terminates (so does the restricted chase of $R$ and $I_0$). This is because the existential variable $z$ in the rule head is instantiated to the skolem term $f_z(b)$, which is passed to variable $y$ in the body atom $P(x,y)$.  As the skolem term $f_z(b)$ is fresh, no trigger to $P(x,y)$ may be available right after the application of $r$. 
\end{exm}

Note that $r$ in Example \ref{exm:counter} depends on itself based on the classic notion of unification. 
To rule out similar false dependencies, we consider a dependency relation under which the cycle $(r,r)$ in the above example is not identified as a dangerous one.
Towards this goal, let us recall the notion of rule dependencies \cite{baget2004improving}\footnote{which was provided earlier for the definition of \aGRD~in Section \ref{previous}} and contrast it with its strengthened version for this section.

\begin{defn}\label{depends-on}
Let $r$ and $r'$ be two arbitrary rules. Recall that $sk(r)$ and $sk(r')$ denote their skolemizations.
\begin{itemize}
\item [(i)] Given an instance $I$, we say that $r'$ {\em depends on $r$ w.r.t.\,$I$} if there is a homomorphism $h: var_u(r)\rightarrow term(I)$ and a homomorphism $g:var_u(r') \rightarrow term(I)\cup term(h(head(sk(r))))$, such that $g(body(r'))\not\subseteq I$.
\item [(ii)] We say that $r'$ {\em depends on $r$} if there is an instance $I$ such that $r'$ depends on $r$ w.r.t. $I$.
\end{itemize}

If the condition in (ii) is not satisfied, we then say that $r'$ {\em does not depend on} $r$, or there is {\em no dependency from $r'$ to $r$}.
Similarly, if the condition in (i) is not satisfied, we then say that $r'$ {\em does not depend on $r$ w.r.t. $I$}, or there is {\em no dependency from $r'$ to $r$ w.r.t. $I$}.
\end{defn}

The definition in (ii) is adopted directly from \cite{baget2004improving}, which is what a general notion of rule dependency is expected, independent of any instance: $r'$ depends on $r$ if there is a way to apply $r$ to derive some new atoms that are used as part of a trigger to $r'$.  That $g$ is not a homomorphism from $body(r')$ to $I$ requires at least one new atom derived by $r$, given $I$.  Since instance $I$ can be arbitrary while satisfying the stated condition, no dependency from $r'$ to $r$ means that no matter what the initial database is and what 
the sequence of derivations is, up to the point of applying $r$, such $I$ that satisfies the stated condition does not exist.

By employing an extended notion of unification, the notion of {\em piece-unification} allows removal of 
a large number of $k$-cycles as irrelevant. We will discuss the details in Section \ref{experiments} when we present our experimentation.

The technical focus of rule dependency in this section is the definition in (i), which is strengthened from (ii) by fixing instance $I$. 
This is needed because our simulations of the restricted chase are generated from some particular, fixed databases. 

Next, we extend the relation of rule dependency to a (non-reflexive) transitive closure. This is needed since a termination analysis often involves sequences of derivations where rule dependencies yield a transitive relation. Given a path $\pi = (r_1, \dots, r_n)$, we are interested in a {\em chain} of dependencies among rules in $\pi$ such that the derivation with $r_n$ ultimately depends on a derivation with $r_1$, possibly via some derivations from rules in between. As a chase sequence may involve independent derivations from other rules in between, in the following, we define the notion of projection to reflect this.

Terminology:
Given a tuple $V = (v_1, \dots, v_n)$ ($n \geq 2$), a {\em projection of $V$ preserving end points}, denoted 
$V' = (v'_1, \dots, v'_m)$, is a projection of $V$ (as defined in usual way), with the additional requirement that the end points are preserved (i.e., $v'_1 = v_1$ and $v'_m=v_n$). By abuse of terminology, $V'$ above will simply be called a {\em projection} of $V$. 

\begin{defn} \label{def:chained}
	Let $R$ be a rule set, $\pi=(r_1, \dots, r_n)$ ($n \geq 2$) a path, and $I_0$ a database. 
Suppose ${\cal I} : I_0, I_1, \dots, I_n$ is a sequence of instances and $H = (h_1, \dots, h_n)$ is a tuple of homomorphisms such that $I_{i-1} \langle r_i, h_i\rangle I_i$ ($1 \leq i \leq n $).
${\cal I}$ is called {\em chained} for $\pi$ if there exists a projection 
${\cal I}' : I_0, I'_1, \dots, I'_m$ of ${\cal I}$, along with the corresponding projections $H'= (h'_1, \dots, h'_m)$ of $H$ and $\pi' = (r'_1, \dots, r'_m)$ of $\pi$, such that for all $1\leq i < m$, $r'_{i+1}$ depends on $r'_i$ w.r.t.\,$I$, where $I = I_0$ if $i = 1$ and $I = I'_{i}\setminus h'_i(head(sk(r'_{i})))$, otherwise.
If ${\cal I}$ is chained for $\pi$, we also say that ${\cal I}$ has the {\em chained property}; for easy reference, we sometime also associate the chained property with the corresponding $H$ and say $H$ is chained, or $H$ is a chained tuple of homomorphisms, w.r.t.\,$I_0$.
\end{defn}

Note that in the definition above, by $I = I'_{i}\setminus h'_i(head(sk(r'_{i})))$, the triggering atoms to $r'_{i+1}$ must include at least one new head atom derived from $r'_i$.

We now address the issue of which databases to check against for termination analysis of the restricted chase.
	For this purpose, let us define a mapping
	$e_i: \mathsf{V}\cup \mathsf{C} \rightarrow \langle\mathsf{V},i\rangle\cup \mathsf{C}$,
	where constants in $\mathsf{C}$ are mapped to themselves and each variable $v\in \mathsf{V}$ is mapped to $\langle v,i\rangle$. 

\begin{defn}\label{def:criticalDB}
	Given a path $\pi=(r_1,r_2,\dots,r_n)$ of a simple rule set, we define:
	$I^{\pi} = \{e_i(body(r_i)): 1\leq  i < n + 1\}$, which is called a {\em restricted critical database of $\pi$}.
\end{defn}

A pair $\langle x,i\rangle$ in $I^{\pi}$ is intended to name a {\em fresh constant} to replace variable $x$ in the body of a rule $r_i$. The atoms in $I^{\pi}$ that are built from these pairs and the constants already appearing in the body of a rule are independent of any given database. The goal is to use these atoms to simulate triggering atoms when necessary, in a derivation sequence from a given database. 
Let us call these pairs {\em indexed constants} and atoms with indexed constants {\em indexed atoms}.
Let us 
use the shorthand $v^i$ for $\langle v, i \rangle$. 

Note that due to the structure of $I^{\pi}$, a trigger for each rule in $\pi$ is automatically available and therefore, without the notion of chained property, a path can rather trivially lead to a restricted chase sequence. To see this, 
we can construct a restricted chase sequence $I_0,I_1,\dots,I_n$ based on $R$ as follows. For each $1\le i\le n$, we construct a trigger $(r_i,h_i)$, where for each variable $v\in var(body(r_i))$, $h_i: v\rightarrow \langle v,i\rangle$. Since indexed constants are fresh, such a trigger is active.

\begin{exm}
Consider the rule set $R$ of Example \ref{exm:counter} and a path $\pi=(r,r)$. For this rule set we have: $I^{\pi}=\{T(x^1, y^1),$ $ P(x^1, y^1),T(x^2, y^2), P(x^2, y^2)\}$.  We see that there does not exist any chained tuple of homomorphisms for $\pi$ w.r.t.\,$I^\pi$. 
In fact, the claim holds for any instance $I$ since there is no rule dependency from $r$ to $r$ (cf. Definition \ref{depends-on}).
\end{exm}

In a restricted critical database that we have seen so far, each body variable is bound to a distinct constant indexed in the order in which rules are applied. 
Later on, we will motivate and introduce the notion of updated restricted critical databases, where distinct indexed constants may be collapsed into the same indexed constant. 

\subsection{Activeness for Simple Rules}
We are ready to define the notion of activeness and show its role in termination analysis for simple rules. 

\begin{defn} ({\bf Activeness})
	\label{defn:activeness}
	Let $R$ be a rule set and $I_0$ a database. 
	A path $\pi=(r_1 \dots,r_n)$ based on $R$ is said to be {\em active} w.r.t.\,$I_0$, if there exists a chained restricted chase sequence ${\cal I}:I_0, \dots, I_n$ for $\pi$. 
\end{defn}

The activeness of a path $\pi$ requires two conditions to hold. First, $\pi$ must lead to a restricted chase sequence and second, the sequence must have the chained property.
In other words, if $\pi$ is not active w.r.t.~$I_0$, either some rule in $\pi$ does not apply due to lack of an active trigger, or the last rule in $\pi$ does not depend on the first in $\pi$ transitively in all possible derivations from $I_0$ using rules in $\pi$ in that order.

Our goal is to simulate a given chained restricted chase sequence w.r.t.\,an arbitrary database by a chained restricted chase sequence w.r.t.\,some fixed databases, while preserving rule dependencies. Such a simulation is called {\em  tight} or {\em dependency-preserving}. For presentation purposes, we will present the results in two stages, first for simple rules for which the restricted critical database $I^{\pi}$ for a path $\pi$ is sufficient. Then, in the next subsection, we present the result for arbitrary rules using updated restricted critical databases. 

\begin{thm} \label{key}
	Let $R$ be a rule set with simple rules and $\pi=(r_1,\dots,r_n)$ a path based on $R$.
	Then, $\pi$ is active w.r.t. some database if and only if $\pi$ is active w.r.t.\,the restricted critical database $I^{\pi}$.
\end{thm}

\begin{proof}
	($\Leftarrow$) Immediate since $I^{\pi}$ is such a database.
	\\
	($\Rightarrow$)
	Let $I$ be a database w.r.t. which $\pi$ is active, i.e., there exists a chained tuple of homomorphisms $H=(h_1,\dots,h_n)$ for $\pi$ such that $(r_i,h_i)$ ($0<i\le n$) is an active trigger on $I_{i-1}$ and $I_{i-1}\langle r_i,h_i\rangle I_i$. So, there exists a sequence
	\begin{eqnarray}
	{\cal A} : I = I_0, I_1,  \dots, I_n
	\end{eqnarray}
	satisfying the condition: for all $1\le i\le n$, 
there is a homomorphism $h_i:var_u(r_i)\rightarrow term(I_{i-1})$, where $r_i\in R$, such that
	\begin{eqnarray}
	h_i(body(r_i)) \subseteq I_{i-1}, %,  \mbox{and} 
	\label{eqn:active-1}~~~~~~~\\
	\forall h_i' \supseteq h_i: h_i'(head(r_i)) \nsubseteq I_{i-1}, ~ \mbox{and} \label{eqn:active-2}
	\\
	I_i = I_{i-1} \cup h_i'(head(r_i)). \label{eqn:active-3}
	%\\
	%(h_1,\dots,h_n) ~ \mbox{is chained for $\pi$} \label{eqn:active-4}
	\end{eqnarray}
	We will construct a chained restricted chase sequence of $R$ w.r.t. $I^{\pi}$ based on a simulation of derivations in $\cal A$.
Let us denote this sequence by 
	\begin{eqnarray}
	{\cal B} : I^{\pi} = I^{\ast}_0, I^{\ast}_1, \dots, I^{\ast}_n.
	\end{eqnarray}
	Then, we need to have properties (\ref{eqn:active-1}), (\ref{eqn:active-2}) and (\ref{eqn:active-3}) 
	%and (\ref{eqn:active-4}) 
	for $\mathcal{B}$ with $h_i$ and $I_{i-1}$ replaced by some homomorphism $g_i$ and instance $I^{\ast}_{i-1}$ respectively, for all $1\le i\le n$.
	
	To show the existence of such a sequence $\cal B$, we show how to construct a tuple of homomorphisms $G = (g_1, g_2, \dots, g_n)$ inductively, such  that $I^{\ast}_{i-1} \langle r_i, g_i\rangle I^{\ast}_i$, for all $1 \leq i \leq n$.
	This ensures that $\cal B$ is a skolem chase sequence. We will then show that
all the triggers are active, and along the way, show that $G$ is a chained sequence. We then conclude that $\cal B$ is, in fact, a chained restricted chase sequence.

Note that instances $I_i$ contain constants from the given database $I$ and instances $I^{\ast}_i$ contain indexed constants. Both may contain some constants appearing in rules in $\pi$.

We construct $g_i$ along with the construction of a many-to-one function $h$ that maps indexed constants appearing in $g_i$ to constants appearing in $h_i$. This provides a relation between $g_i$ and $h_i$. For any atom $a \in body(r_i)$, we call atom $h_i(a)$ an {\em image} of $g_i(a)$.  The function $h$ is many-to-one because distinct indexed constants in $g_i$ may need to be related to a constant in $h_i$ in simulation (in generating sequence ${\cal A}$, distinct variables may be bound to the same constant; but in generating sequence $\cal B$, distinct variables can only be bound to distinct indexed constants).

For $i=1$, we let $g_1(body(r_1)) \subseteq I^\pi$ with the index in indexed constants being $1$.  Such  $g_1$ uniquely exists. 
 As $(r_1, g_1)$ is clearly a trigger, we 
have $I^{\ast}_0 \langle r_1, g_1\rangle I^{\ast}_1$ under the skolem chase. For function $h$, clearly we can let $h$ be such that $h(g_1(a)) =h_1(a)$ for each atom $a \in body(r_1)$.

For any $1 < i \leq n$, we construct $g_i$ as follows. Let $a \in body(r_i)$. If $h_i (a) \in I$, i.e., the triggering atom $h_i(a)$ is from database $I$, then we let
$g_i$ map $a$ to the corresponding indexed atom in $I^\pi$ with index $i$.
If $h_i (a) \not \in I$, i.e., $h_i(a)$ is a derived atom, we then let $g_i(a)$ be any atom whose image is $h_i(a)$.\footnote{Recall that $h$ is in general many-to-one. So, we may have multiple atoms whose image is $h_i(a)$. Since the rules are assumed to be simple, choosing any of these atoms can lead to the construction of a desired tuple of homomorphisms $G$ as well as the function $h$.} Then, we can extend function $h$ by $h(g_i (a)) = h_i(a)$. Note that this is always possible for simple rules since $body(r_i)$ has no repeated variables. By construction, that $(r_i,h_i)$ is a trigger on $I_{i-1}$ implies that $(r_i,g_i)$ is a trigger on $I^{\ast}_{i-1}$.

	We now show that all triggers $(r_i, g_i)$ $(1 \leq i \leq n)$ are active, i.e.,
	\begin{eqnarray}
	\forall g_i' \supseteq g_i: g_i' (head(r_i))\nsubseteq I_{i-1}^{\ast}   \label{eqn:active},~~ 1 \leq i \leq n
	\end{eqnarray}
	To relate homomorphisms $g_i$ with $\cal B$ to $h_i$ with $\cal A$, from above we have $h(g_i(x)) = h_i(x)$, for all $x \in var_u (r_i)$. Then, it follows 
	\begin{eqnarray} 
	h(I_{i-1}^{\ast}) \subseteq I_{i-1},~~1 \leq i \leq n
	\label{eqn:mapping-subset}
	\end{eqnarray}
	which can be shown by induction: for the base case, we have $h(I_0^{\ast}) \subseteq I_0$ by definition, and for the induction step, for each $k \geq 1$, that $h(I_{k-1}^{\ast}) \subseteq I_{k-1}$ implies $h(I_{k}^{\ast}) \subseteq I_{k}$ is by the construction of homomorphism $g_k$ in $\cal B$.
	
	To prove (\ref{eqn:active-2}), for the sake of contradiction, assume that
	it does not hold,  i.e., $\exists g_i' \supseteq g_i$ s.t. $g_i' (head(r_i))\subseteq I_{i-1}^{\ast}$. This together with (\ref{eqn:mapping-subset}) implies  $h(g_i' (head(r_i)))\subseteq h(I_{i-1}^{\ast})\subseteq I_{i-1}$.
	Now let $h_i' (x)=h(g_i' (x))$. It follows 
	$h_i' (head(r_i))\subseteq I_{i-1}$, a contradiction to (\ref{eqn:active-2}). Therefore, all triggers applied in $\cal B$ are active and $\pi$ thus leads to a restricted chase sequence of $R$ and $I^{\pi}$.

Finally, ${\cal B}$ is chained because the {\em depends-on} relation in ${\cal A}$ is preserved for ${\cal B}$.
For the path $\pi=(r_1,\dots,r_n)$, assume that $r_j$ depends on $r_i$ w.r.t. $I_{i-1}$ $(1 \leq i < j \leq n)$. As ${\cal A}$ is a restricted chase sequence, 
we have homomorphisms $h_i: body(r_i) \rightarrow I_{i-1}$ and $h_j: body(r_j) \rightarrow I_{j-1}$. That $r_j$ depends on $r_i$ w.r.t. $I_{i-1}$ 
ensures that 
$h_j$ is not a homomorphism from $body(r_j)$ to $I_{j-1} \setminus h_i(head(sk(r_i)))$. 
We have already shown the existence of homomorphisms  $g_i:  body(r_i) \rightarrow I^{\ast}_{i-1}$ and $g_j:  body(r_j) \rightarrow I^{\ast}_{j-1}$. 
Since $h_j$ is not a homomorphism from $body(r_j)$ to $I_{j-1}\setminus h_i(head(sk(r_i)))$, it follows by construction that 
$g_j$ is not a homomorphism from $body(r_j)$ to $I^{\ast}_{j-1} \setminus  g_i(head(sk(r_i)))$.  We, therefore, conclude that $r_j$ depends on $r_i$ w.r.t. $I^{\ast}_{i-1}$ $(1 \leq i < j \leq n)$.
We are done.
\end{proof} 

\subsection{Activeness for Arbitrary Rules}
For non-simple rules, a tight simulation using the restricted critical database $I^{\pi}$ for a given path $\pi$ is not always possible. The following example demonstrates that not all active paths can be simulated.
	\begin{exm}\label{exm:counter2}
	Consider the following rule set $R=\{r_1,r_2,r_3\}$, where
	$$\begin{array}{ll}
	r_1: P(x,y)\rightarrow Q(x,y)\\
	r_2: R(x,y)\rightarrow T(x,y)\\
	r_3: Q(x,y), T(x,y)\rightarrow \exists z \, P(z,x), R(z,x)\\
	\end{array}$$
$R$ is not all-instance terminating since for database $I = \{P(a,b), R(a,b)\}$, there is a non-terminating restricted chase sequence starting from $I$ (assuming that the existential variable $z$ is skolemized to $f_z(x)$):%
$$\begin{array}{ll}
I_0 = I ~~~~~~~~~~~~~~~~~~~~~~~~~~
I_1 = I_0 \cup \{Q(a,b)\} ~~~~~~~\\
I_2 = I_1 \cup \{T(a,b)\} ~~~~~~~
I_3 = I_2 \cup \{P(f_z(a),a), R(f_z(a),a)\} ~~~\\......
\end{array}
$$
where the corresponding active triggers $(r_1,h_1), (r_2,h_2), (r_3,h_3)$ can be easily identified. 
However, as illustrated below, a tight simulation for any path $\pi = (r_1,r_2, \dots)$ is not possible for the restricted critical database $I^{\pi}$. 
For example, given $\pi_1 =(r_1,r_2,r_3)$, with restricted critical database $I^{\pi_1} = \{P(x^1,y^1), R(x^2,y^2),Q(x^3,y^3), T(x^3,y^3)\}$, it is easy to verify that $\pi_1$ is not active w.r.t. $I^{\pi_1}$.
To see why this is the case, consider the following derivation which is obtained after having applied the triggers $(r_1,g_1)$ and $(r_2,g_2)$ to produce
$$\begin{array}{ll}
I^{\ast}_0 = I^{\pi_1}, ~~~~~~
I^{\ast}_1 = I^{\ast}_0  \cup \{Q(x^1,y^1)\}, ~~~~~~
I^{\ast}_2 = I^{\ast}_1 \cup \{T(x^2,y^2)\}
\end{array}
$$
The reason that $\pi_1$ is not active w.r.t. $I^{\pi_1}$ is that multiple occurrences of constants $a$ and $b$ in the triggering atoms on $I_2$, i.e., $Q(a,b)$ and $T(a,b)$, are originated from the given database from {\em different} sources (atoms).  
For termination analysis, we must provide a simulation of any restricted chase sequence. Below, we discuss two possible solutions using the above example.

\begin{itemize}
\item Solution 1: Trigger $(r_3, \{x/x^3, y/y^3\})$ on $I^{\ast}_2$ is already available since $Q(x^3,y^3), T(x^3,y^3) \in I^{\ast}_2$, which can be applied to continue the chase.
\item Solution 2: Let $rn$ be a renaming function that renames indexed constants $x^2$ and $y^2$ appearing in $I^{\pi_1}$ to $x^1$ and $y^1$ respectively, i.e., $rn(I^{\pi_1}) =  \{P(x^1,y^1), R(x^1,y^1), Q(x^3,y^3), T(x^3,y^3)\}$, so that $(r_3, \{x/x^1, y/y^1\})$ is a trigger on $rn(I^{\ast}_2)$.
\end{itemize}
\end{exm} 
%}

Solution 1 is rather weak since it allows the simulation of a chained sequence to be ``broken" without preserving rule dependency, whereas Solution 2 leads to a tight simulation, i.e., a simulation that preserves the dependency relation of the sequence being simulated.
In this paper, we formalize and develop results for Solution 2. 

Given a path $\pi$ and critical database $I^\pi$, let $\Pi_{I^\pi}$ be the set of indexed constants appearing in $I^\pi$. We define 
a {\em renaming function for $I^\pi$} to be a mapping from $\Pi_{I^\pi}$ to  $\Pi_{I^\pi}$. 
For technical clarity, we eliminate symmetric renaming functions by imposing a restriction: an indexed constant with index $i$ 
can only be renamed to an indexed constant with index $j$, where $1 \leq j < i$. In other words, an indexed constant with index $i$ in $I^\pi$ can only be renamed to one which appears in a rule in $\pi$ earlier than $r_i$.

\begin{thm} \label{key0}
	Let $R$ be a rule set and $\pi=(r_1,\dots,r_n)$ a path based on $R$.
	Then, $\pi$ is active w.r.t.\,some database if and only if there exists a renaming function $rn^{\ast}$ for $I^\pi$ such that $\pi$ is active 
w.r.t. $rn^{\ast}(I^\pi)$, where $rn^{\ast}$ is composed of at most $n$ renaming functions.
\end{thm}

\begin{proof}
($\Leftarrow$) Immediate since $rn^{\ast}(I^\pi)$ is such a database.
\\
($\Rightarrow$)
The proof follows the same structure as for Theorem \ref{key} except for the case where the tight simulation of a chase step fails to provide a trigger due to repeated variables in a rule body.

As in the proof of Theorem \ref{key}, we assume that path $\pi = (r_1, \dots, r_n)$ is active w.r.t.\,some database $I$, so that there is a chained restricted chase sequence
\begin{eqnarray}
	{\cal A} : I = I_0, I_1,  \dots, I_n
	\end{eqnarray}
generated by active triggers $(r_1, h_1), \dots, (r_n, h_n)$. We show that there exist a renaming function $rn^{\ast}$ for $I^\pi$ and a chained restricted chase sequence w.r.t. $rn^{\ast}(I^{\pi})$
	\begin{eqnarray} \label{B}
	{\cal B} : rn^{\ast}(I^{\pi}) = rn^{\ast}(I^{\ast}_0), rn^{\ast}(I^{\ast}_1), \dots, rn^{\ast}(I^{\ast}_n).
	\end{eqnarray}
generated by active triggers 
$(r_1, rn^{\ast}\circ g_1), \dots, (r_n, rn^{\ast}\circ g_n))$. We prove the existence of ${\cal B}$ by constructing $g_i$'s (and its renamed counterparts) along with the construction of a many-to-one function $h$ that relates indexed constants in $g_i$ (and its renamed counterparts) to constants in $h_i$. We apply the same argument repeatedly to show the existence of a composed renaming function $rn^{\ast}$. Let us start by constructing the first renaming function, $rn_1$.

The construction of $g_1$ is the same as in the proof of Theorem \ref{key} \-- we let $g_1(body(r_1)) \subseteq I^\pi$ with the index in indexed constants being $1$ and 
let $h(g_1(body(r_1))) =h_1(body(r_1))$.
For the inductive case ($1 < i \leq n$), we construct $g_i$ as follows. Let $a\in body(r_i)$. If $h_i (a)\in I$, i.e., the triggering atom $h_i(a)$ is from database $I$, then we let $g_i$ map $a$ to a corresponding indexed atom in $I^\pi$ with index $i$.
If $h_i (a)\not \in I$, i.e., $h_i(a)$ is a derived atom, we then consider all body atoms of $r_i$ including $a$ that form a {\em connected component} in that any two of which share at least one variable. There are in general one or more such connected components in $body(r_i)$.
For simplicity and w.l.o.g., let us assume that $body(r_i)$ consists of only one such connected component.
If $body(r_i)$ for some $1\le i\le n$ consists of more than one connected component, then we can apply the same techniques used below to construct a sequence of renaming functions \-- as long as the required properties for the construction of these functions are met for each component (cf. Case (ii) below), the same argument is applicable.

Now let us attempt to construct a mapping $g_i$ by letting $g_i(body(r_i))$ be the set of atoms whose images are precisely those in $h_i(body(r_i))$. 
There are two cases.

\smallskip \smallskip
\noindent
{\bf Case (i)} $g_i$ is a homomorphism from $body(r_i)$ to $I^{\ast}_{i-1}$. In this case, function $h$ can be extended by $h(g_i (body(r_i))) = h_i(body(r_i))$. By construction, $(r_i, g_i)$ is a trigger on $I^{\ast}_{i-1}$, and the proof that $(r_i,g_i)$ %(in fact, all triggers $(r_j,g_j)$ $(1 \leq j \leq i)$) 
is active remains the same as for Theorem \ref{key}.

\smallskip \smallskip
\noindent
{\bf Case (ii)} Otherwise $g_i$ fails to be a homomorphism from $body(r_i)$ to $I^{\ast}_{i-1}$. Assume $g_i$ is the first such failure in the construction of sequence $\cal B$ so far. Note that the failure is because $g_i$ constructed this way must be a one-to-many mapping \-- $g_i$ must map a variable to distinct indexed constants because
multiple occurrences of a variable in $body(r_i)$ are instantiated to a common constant by $h_i$  but to simulate that, $g_i$ must map the same variable to distinct indexed constants. 

The failure can be remedied by a renaming function for $I^\pi$, denoted $rn_1$, by which some different indexed constants are renamed to the same one
so that $(r_i, rn_1 \circ g_i)$
 is a trigger on $rn_1(I^{\ast}_{i-1})$. Clearly, such a renaming function exists. 
We require that $rn_1$ be {\em minimal} in that the number of indexed constants that are renamed to {\em different ones} is minimized.\footnote{In other words, that an indexed constant is renamed to a different one only when it is necessary.} It is easy to see that the existence of a renaming function for $I^\pi$ implies the existence of such a minimal renaming function for $I^\pi$.  We now want to show that the sequence
\begin{eqnarray}\label{rn1}
	rn_1(I^{\pi}) = rn_1(I^{\ast}_0), rn_1(I^{\ast}_1), \dots, rn_1(I^{\ast}_{i-1}), rn_1(I^{\ast}_i)
	\end{eqnarray}
is a chained restricted chase sequence generated by triggers $(r_1, rn_1 \circ g_1), \dots, (r_i, rn_1 \circ g_i)$. The function $h$ that relates indexed constants to constants in $h_j$ $(1 \leq j \leq i)$ is updated correspondingly as  $h(rn_1 \circ g_j(body(r_j))) =h_j(body(r_j))$.

That $(r_i, rn_1 \circ g_i)$ is a trigger on $rn_1(I^{\ast}_{i-1})$ is by the construction of $rn_1$.
For each $rn_1 \circ g_j$ $(1\leq j < i)$, since for $rn_1 \circ g_j$ the only update of $g_j$ is that some different indexed constants are replaced by the same one; that $g_j$ is a homomorphism from $body(r_j)$ to $I^{\ast}_{j-1}$ implies that
$rn_1 \circ g_j$ is a homomorphism from $body(r_j)$ to $rn_1(I^{\ast}_{j-1})$. 
We now show that triggers $(r_j, rn_1 \circ g_j)~(1\leq j \leq n)$ are all active. 

The intuition behind this part of the proof is that in case (i) when we use distinct indexed constants for distinct variables, we do not introduce any possibility of ``recycled" atoms (i.e., atoms which can also be used in later derivations). Therefore, the activeness of $(r_j, h_j)$ implies activeness of $(r_j,g_j)$. On the other hand, although the above statement may not hold for case (ii), a renaming function that is minimal ensures that we do not introduce more than what is needed, i.e., $rn_1\circ g_j$ requires no more mappings to the same constants than $h_j$. This again ensures that the activeness of trigger $(r_j, h_j)$ implies activeness of trigger $(r_j, rn_1\circ g_j)$. 

More formally, the activeness of $(r_j, rn_1 \circ g_j)~(1 \leq j \leq n)$ means that the following conditions hold: for each $1 \leq j \leq n$

	\begin{eqnarray}
	\forall g_j' \supseteq rn_1 \circ g_j: g_j' (head(r_j))\nsubseteq rn_1(I_{j-1}^{\ast}) \label{eqn:active},~~ 1 \leq j \leq i
	\end{eqnarray}
		
	We let $h(rn_1 \circ g_j(x)) = h_j(x)$, for all $x \in var_u (r_j)$. Then, by induction we show that
	\begin{eqnarray} 
	h(rn_1(I_{j-1}^{\ast})) \subseteq I_{j-1},~~1 \leq j \leq i
	\label{eqn:mapping-subset}
	\end{eqnarray}

	For the base case, we have $h(rn_1(I_0^{\ast})) \subseteq I_0$, which holds due to the minimality of $rn_1$. For the induction step, for each $k \geq 1$, let us assume $h(rn_1(I_{k-1}^{\ast})) \subseteq I_{k-1}$. Then we need to show $h(rn_1(I_{k}^{\ast})) \subseteq I_{k}$. The latter can be done by the construction of $rn_1\circ g_k$ in (\ref{rn1}).
	
	We then proceed to prove (\ref{eqn:active}). For this purpose, assume that it does not hold,  i.e., $\exists g_j' \supseteq rn_1\circ g_j$ s.t. $g_j' (head(r_j))\subseteq rn_1(I_{j-1}^{\ast})$. This together with (\ref{eqn:mapping-subset}) implies  $h(g_j' (head(r_j)))\subseteq h(rn_1(I_{j-1}^{\ast}))\subseteq I_{j-1}$.
	Now let $h_j' (x)=h(g_j' (x))$. It follows 
	$h_j' (head(r_j))\subseteq I_{j-1}$, 
	which is a contradiction to our assumption that $(r_j, h_j)$ for $1\le j\le i$ active.
	This shows that all triggers applied in (\ref{rn1}) are active.

We then apply the same argument to continue the construction of sequence $\cal B$ of (\ref{B}): 
\begin{eqnarray}
	rn_1(I^{\pi}) = rn_1(I^{\ast}_0), rn_1(I^{\ast}_1), \dots, rn_1(I^{\ast}_{i-1}), rn_1(I^{\ast}_i), \dots
	\end{eqnarray}
generated by active triggers $ (r_j, rn_1 \circ g_j)$ $(1 \leq j \leq i)$ from the updated restricted critical database $rn_1(I^{\pi})$.
If case (i) applies for the simulation of a chase step in $\cal A$, then let us use the identity renaming function (which is minimal by definition). Thus, the simulation of each chase step results in a minimal renaming function. 
It follows that $rn^{\ast} = rn_n \circ \dots \circ rn_1$ and, as the chained property immediately holds by construction, sequence $\cal B$ is indeed a chained restricted chase sequence. We then conclude that $\pi$ is active w.r.t. the updated restricted critical database $rn^{\ast}(I^\pi)$. 
We are done. 
\end{proof}

In the sequel, given a path $\pi$, $I^\pi$ and $rn^{\ast}(I^\pi)$ for all renaming functions $rn^{\ast}$ are all called a restricted critical database. For clarity, we may qualify the latter as an updated restricted critical database.

The development of this section leads to the following conclusion, which can be considered the foundation of our approach to defining classes of the finite restricted chase in the paper. 

\begin{thm}\label{terminate}
Let $R$ be a rule set. For any $k > 0$, if
no $k$-cycle $\sigma$ is active w.r.t.\,$rn^{\ast}(I^\sigma)$, %where 
for all renaming functions $rn^{\ast}$ for $I^\sigma$, then $R$ is all-instance terminating under the restricted chase. 
\end{thm}
\begin{proof}
Assume that $R$ is not all-instance terminating under restricted chase. Then for some database $I_0$ there is a non-terminating restricted chase sequence ${\cal I}: I_0,  \dots , I_j, ... $.
Since $I_0$ is finite, there can only be a finite number of independent applications of any rule. It follows that ${\cal I}$ must contain one chained restricted chase sequence for some $k$-cycle $\sigma$. 
W.l.o.g., assume that $\sigma$
appears immediately after an initial, finite segment of $\cal I$, say $I_0, \dots, I_{i}$. It follows that the non-terminating sequence ${\cal I}$ without this initial finite segment ${\cal I'}: I_i, \dots, I_j, ...$ is a non-terminating chained restricted chase sequence.

By the contraposition of the only if statement of Theorem \ref{key0}, the assumption that $\sigma$ is not active w.r.t.\,$rn^{\ast}(I^\sigma)$ for all renaming function $rn^{\ast}$ for $I^\sigma$, implies that $\sigma$ is not active w.r.t. any database, i.e., a chained restricted chase sequence for $\sigma$ does not exist, for any database, which results in a contradiction.
\end{proof}

As we have seen up to this point that renaming enables a tight simulation for termination analysis based on testing $k$-cycles. A question is whether renaming is a necessary condition in general for our termination analysis. The question is raised due to the following observation. 
\begin{exm}\label{exm:counter2_observation} Consider Example \ref{exm:counter2} again. We have seen that path $\pi_1=(r_1, r_2, r_3)$ requires renaming in order to obtain a tight simulation. Now consider $\pi_2=(r_3, r_2, r_1)$, which is a permutation of $\pi_1$. It can be shown that unlike $\pi_1$ which is not active w.r.t. restricted critical database $I^{\pi_1}$, 
$\pi_2$ is active w.r.t. restricted critical database $I^{\pi_2}$.
According to Theorem \ref{terminate}, as long as there is one $k$-cycle that is active, we do not conclude that the given rule set is all-instance terminating. For this example, since the 1-cycle $\sigma =(r_3, r_2, r_1, r_3)$ 
is active w.r.t. the restricted critical database $I^{\sigma}$, there is no conclusion that $R$ is all-instance terminating.  This may suggest that if we test all $k$-cycles, the mechanism of renaming may be redundant.
However, the next example shows that this is not the case in general. 
\end{exm}

	\begin{exm}\label{exm:counter3}
		Consider the following rule set $R' =\{r_1,r_2,r_3\}$ modified from rule set $R$ of Example \ref{exm:counter2}, where
		$$\begin{array}{ll}
		r_1: P(x,y,z), K(z)\rightarrow Q(x,y,z)\\
		r_2: R(x,y,z)\rightarrow T(x,y,z)\\
		r_3: Q(x,y,z), T(x,y,z)\rightarrow \exists v \, P(v,x,z), R(v,x,z)\\
		\end{array}$$
	$R'$ is not all-instance terminating which can be verified using the database $$I_0 = \{P(a,b,c), K(c), R(a,b,c)\}$$ 
We have the following chase sequence starting from database $I_0$ (assuming that $f_v$ is used to skolemize the existential variable $v$) by applying the rules in the path $(r_1,r_2,r_3)$ repeatedly.
	$$\begin{array}{ll}
	I_1 = I_0 \cup \{Q(a,b,c)\}, ~~~~
	I_2 = I_1 \cup \{T(a,b,c)\},\\
	I_3 = I_2 \cup \{P(f_v(a,c),a,c), R(f_v(a,c),a,c)\}, ~~~
	I_4 = I_3 \cup \{Q(f_v(a,c),a,c)\}, \\......
	\end{array}
	$$
The question is: by testing all 1-cycles, can we capture this non-terminating behavior without using renaming? As we show below, the answer is negative.

	Similar to the rule set of Example \ref{exm:counter2}, a tight simulation is not possible for any path of the form $\pi=(r_1,r_2, \dots)$ w.r.t. the restricted critical database $I^\pi$. However, unlike the rule set of Example \ref{exm:counter2}, no permutation $\pi'$ of $\pi$ may lead to a tight simulation for $\pi'$ w.r.t. the restricted critical database $I^{\pi'}$.
For example, consider the path $\pi_2=(r_3, r_2, r_1)$ which is a permutation of $\pi_1=(r_1, r_2, r_3)$.
	The restricted critical database of $\pi_2$ is as follows:
	$$I^{\pi_2} = \{Q(x^1,y^1,z^1), T(x^1,y^1,z^1), R(x^2,y^2,z^2), P(x^3,y^3,z^3), K(z^3)\}$$ 
	and we derive the following restricted chase sequence:
	\begin{eqnarray} \label{pi2}
	I^{\ast}_0 = I^{\pi_2}, ~~~~~
	I^{\ast}_1 = I^{\ast}_0  \cup \{P(f_v(x^1,z^1), x^1, z^1), R(f_v(x^1,z^1), x^1, z^1)\}, ~~~~~~\\ \nonumber
	I^{\ast}_2 = I^{\ast}_1 \cup \{T(f_v(x^1,z^1), x^1, z^1)\},~~~~~
	I^{\ast}_3 = I^{\ast}_2 \cup \{Q(x^3, y^3, z^3)\}~~~~~~~
	\end{eqnarray}
	
It is easy to check that after derivation of 	$I^{\ast}_2$, no trigger for $r_1$ exists that uses atoms derived in $I^{\ast}_2$.
Therefore, to derive $I^{\ast}_3$, we have no choice but to pick homomorphism $h=\{x/x^3, y/y^3, z/z^3\}$ to construct trigger $(r_1,h)$ to derive $Q(x^3,y^3,z^3)$. Therefore, the restricted chase terminates since there is no trigger from $I^{\ast}_3$. A similar argument applies to other permutations of $\pi_1$. If we conclude that $R'$ is all-instance terminating based on testing all 1-cycles without renaming, we would get a wrong conclusion.

On the other hand, given a (finite) path $\pi$, if it leads to a chained restricted chase sequence, starting from the updated restricted critical database $rn^{\ast}(I^{\pi})$ for some renaming function $rn^{\ast}$, then there is a tight simulation so that $\pi$ is shown to be active. For example, for $\pi_2=(r_3, r_2, r_1)$ above we can find an updated restricted critical database of $\pi_2$ as follows:
$$rn^{\ast}(I^{\pi_2})=\{Q(x^1, y^1, z^1), T(x^1, y^1, z^1), 
R(x^2, y^2, z^2),
P(x^1,y^1,z^1), K(z^1)\}$$
where indexed constants with index 3 are renamed to those with index 1,
so that $\pi_2$ is active w.r.t. $rn^{\ast}(I^{\pi_2})$.
\end{exm}

\section{$K$-$\mathsf{Safe}(\Phi)$ Rule Sets}\label{k-safe}
We now apply the results of the previous section to define classes of the finite restricted chase.
The idea is to introduce a parameter of {\em cycle function} to generalize various acyclicity notions in the literature, and we will test a path only when it fails to satisfy the given acyclicity condition.

\begin{defn} \label{def:cycle_fun}
	Let $R$ be a rule set and  $\Sigma$ the set of all finite cycles based on $R$. A {\em cycle function} is a mapping $\Phi^R:  \Sigma \rightarrow \{T,F\}$, where $T$ and $F$ denote $true$ and $false$, respectively.
\end{defn}

Let $\Phi$ be the binary function from rule sets and cycles such that
%of which $\Phi^R$ is the projection function on its first parameter, i.e. 
$\Phi(R,\sigma) = \Phi^R(\sigma)$, where $R$ is a rule set and $\sigma$ is a cycle. By overloading, the function $\Phi$ is also called a cycle function. 

We now address the question of how to obtain a cycle function for an arbitrary rule-based acyclicity condition of finite skolem chase e.g., JA \cite{krotzsch2011extending}, aGRD \cite{baget2004improving}, MFA \cite{cuenca2013acyclicity}, etc.

\begin{defn} \label{def:cycle_function_acyclicity}
	Let $\Delta$ denote an arbitrary acyclicity condition of finite skolem chase (for convenience, let us also use $\Delta$ to denote the class of rule sets that satisfy the acyclicity condition expressed by $\Delta$). We define 
	a cycle function $\Phi_\Delta$ as follows: for each rule set $R$ and each cycle $\sigma$ based on $R$, if the acyclicity condition $\Delta$ holds for rules in $\mathsf{Rule}({\sigma})$,\footnote{Recall that $\mathsf{Rule}({C})$ is the set of distinct rules in $C$.} then $\Phi_\Delta$ maps $(R,\sigma)$ to $T$; otherwise  $\Phi_\Delta$ maps $(R,\sigma)$ to $F$.
\end{defn}

That is, $\Phi_\Delta$ maps $(R,\sigma)$ to $T$ whenever the acyclicity condition $\Delta$ for the rule set $\mathsf{Rule}({\sigma})$ is satisfied and to $F$ otherwise. Since any non-terminating restricted chase sequence must involve a cycle of rules, any sufficient condition for acyclicity by definition already guarantees restricted chase termination.

In the sequel, we will use $RS(\Delta)$ to denote the class of rule sets that satisfy the acyclicity condition $\Delta$. 
Also, because of 
Definition~\ref{def:cycle_function_acyclicity}, we will feel free to write 
$\Phi_\Delta(R,\mathsf{Rule}(\sigma))$ for  $\Phi_\Delta(R,\sigma)$.

\begin{exm} \label{exm:cycle_function}
	Consider the rule set $R_1$ from Example \ref{exm:simplified} and assume $\Delta=\text{aGRD}$ in 
Definition~\ref{def:cycle_function_acyclicity}.
Recall that a rule set $R$ belongs to aGRD  (acyclic graph of rule dependencies) if there is no cyclic dependency relation between any two (not necessarily different) rules from $R$, possibly through other dependent rules of $R$. Clearly, the corresponding cycle function $\Phi_{\Delta}$ maps both cycles  $\sigma_1=(r_1,r_2,r_1)$ and $\sigma_2=(r_2,r_1,r_2)$ to $T$.
\end{exm}

We are ready to present our hierarchical approach to defining classes of the finite restricted chase. In the following, we may just write $\Phi$ for $\Phi_{\Delta}$ as a parameter for cycle functions, or as some fixed cycle function, in particular in a context in which an explicit reference to the underlying acyclicity condition $\Delta$ is unimportant. 

\begin{defn}\label{def:k-safe(cyc)}
	{ ({\bf $k$-${\mathsf{safe}(\Phi)}$ rule sets})} 
		Let $R$ be a rule set and $\sigma$ a $k$-cycle $(k\geq 1)$.
		We call $\sigma$ {\em safe} if for all databases $I$, $\sigma$ is not active w.r.t.~$I$.
		Furthermore, $R$ is said to be in $k$-$\mathsf{safe}(\Phi)$, or to belong to $k$-$\mathsf{safe}(\Phi)$ (under cycle function $\Phi$), if for every $k$-cycle $\sigma$ which is mapped to $F$ under $\Phi^R$, $\sigma$ is safe.
\end{defn}

For notational convenience, for $k=0$ we may write $0$-$\mathsf{safe}(\Phi_{\Delta})$ for $RS(\Delta)$. 
For example, it can be verified that the rule set $R_1$ in Example \ref{exm:simplified} is in $k$-$\mathsf{safe}(\Phi_{\Delta})$ for any $k\geq 1$ and any cycle function $\Phi_{\Delta}$ based on some skolem acyclicity condition $\Delta$ in the literature such as weak-acyclicity (WA) \cite{fagin2005data}, Joint-acyclicity (JA) \cite{krotzsch2011extending}, and MFA \cite{cuenca2013acyclicity}, etc.
It is also not difficult to see that the rule set $R_2$ in the same example belongs to $2$-$\mathsf{safe}(\Phi_{\aGRD})$ as well as $2$-$\mathsf{safe}(\Phi_{\WA})$ (but note that they do not belong to
$1$-$\mathsf{safe}(\Phi_{\aGRD})$ or $1$-$\mathsf{safe}(\Phi_{\WA})$).  However, we stress that $R_2$ does not belong to any known class of acyclicity, including RMFA and RJA. That is, rule sets like $R_2$ are recognized as a finite chase only under the hierarchical framework proposed in this paper. 

By Theorem \ref{key0}, $k$-${\mathsf{safe}(\Phi)}$ can be equivalently defined in terms of restricted critical databases. 

\begin{prop} \label{lem:arbitraryDB-renamedDB}
	 For any cycle function $\Phi$, a rule set $R$ is in $k$-${\mathsf{safe}(\Phi)}$ if and only if every $k$-cycle $\sigma$ which is mapped to $F$ under $\phi^R$ is not active w.r.t.\,$rn^{\ast}(I^\sigma)$, for all renaming functions $rn^{\ast}$. 
\end{prop}

We are now in a position to show the following theorem.

\begin{thm} \label{thm:kSAFE}
Let $\Phi_\Delta$ be a cycle function.
For all $k \geq 1$,  $(k-1)$-$\mathsf{safe}(\Phi_{\Delta})\subseteq$ $k$-$\mathsf{safe}(\Phi_{\Delta}) \subseteq \mathsf{CT}_{\forall\forall}^\mathsf{res}$.
\end{thm}

\begin{proof}
	For the first subset relation, let us first consider the base case where $k=1$. Since any non-terminating skolem chase goes through at least one 1-cycle based on $R$, if none of the 1-cycles on $R$ violates the corresponding acyclicity condition, i.e., $\Phi_\Delta$ maps any 1-cycle $\sigma$ to $T$, then $R$ trivially belongs to $RS(\Delta)$. Thus, $RS(\Delta) = 0$-$\mathsf{safe}(\Phi_{\Delta})\subseteq$ $1$-$\mathsf{safe}(\Phi_{\Delta})$.  
Then, for all renaming functions $rn^{\ast}$, if there is no chained restricted chase sequence of $R$ and $rn^{\ast}(I^\sigma)$ for a $k$-cycle $\sigma$, then there is no chained restricted chase sequence of $R$ and $rn^{\ast}(I^{\sigma'})$ for any $(k+1)$-cycle $\sigma'$, since the latter goes through at least one $k$-cycle. This shows the first subset relation.

To show the second subset relation, let $R \in k$-$\mathsf{safe}(\Phi_\Delta)$, for any fixed $k \geq 1$. For all $k$-cycle $\sigma$, if $(R,\sigma)$ is mapped to $T$ by $\Phi_{\Delta}$ for every $k$-cycle $\sigma$, then by definition $R\in \mathsf{CT}_{\forall\forall}^\mathsf{sk}\subset \mathsf{CT}_{\forall\forall}^\mathsf{res}$. If for some $k$-cycle $\sigma$ such that $(R,\sigma)$ is mapped to $F$ by $\Phi_{\Delta}$, then by Proposition \ref{lem:arbitraryDB-renamedDB}, $R\in k$-$\mathsf{safe}(\Phi_\Delta)$ implies that $\sigma$ is not active w.r.t.\,$rn^{\ast}(I^\sigma)$ for all renaming functions $rn^{\ast}$ for $I^\sigma$. It then follows from inactiveness (Definition~\ref{defn:activeness}) and Proposition \ref{lem:arbitraryDB-renamedDB} that there are no chained restricted chase sequences of $R$ and $rn^{\ast}(I^\sigma)$. Thus, by Theorem \ref{terminate}, $R$ is restricted chase terminating w.r.t. $rn^{\ast}(I^\sigma)$.
By the first subset relation, for all $k' > k$, all $k'$-cycles are terminating. Therefore, we have $k$-$\mathsf{safe}(\Phi_{\Delta}) \subseteq \mathsf{CT}_{\forall\forall}^\mathsf{res}$.
\end{proof}

Finally, we present Algorithm \ref{alg:kSafe} to determine whether a rule set belongs to the class $k$-$\mathsf{safe}(\Phi_{\Delta})$. The procedure returns $true$ if it is and $false$ otherwise.
%\comment{
\begin{algorithm}
	\caption{$k$-$\mathsf{safe} ~Algorithm$}\label{alg:kSafe}
	\textbf{Input:} A set of rules $R$; An integer $k\geq 0$; A cycle function $\Phi$\\
	\textbf{Output:} Boolean value $IsAcyclic$;
	\begin{algorithmic}[1]
		\Procedure{$k$-$\mathsf{safe}(R,\Phi)$}{}
		\EndProcedure
		\State $bool$ $IsAcyclic \leftarrow true$;
		\For {$\text{each $k$-cycle } \sigma$ based on $R$}\label{marker}
			\If {$\Phi(R,\mathsf{Rule}(\sigma))=F$}
				\State Find the restricted critical database $I^{\sigma}$;
				\For {$\text{each renaming function } rn^{\ast}$}
					\If {$\sigma$ is active w.r.t.\,$rn^{\ast}(I^{\sigma})$}
						\State \Return $\lnot IsAcyclic$;
					\EndIf
				\EndFor
			\EndIf
		\EndFor
	\State \Return $IsAcyclic$;
\end{algorithmic}
\end{algorithm}

\begin{prop}
Given a rule set $R$, a cycle function $\Phi_\Delta$ and an integer $k \geq 1$, $R$ belongs to $k$-$\mathsf{safe}({\Phi}_\Delta)$ if and only if Algorithm $k$-$\mathsf{safe}$ returns $true$.\footnote{The algorithm can be improved by considering only minimal renaming functions, which however would not lower the complexity upper bound. For this reason, we do not pursue the improvement at this level of abstraction.}
\end{prop}

\begin{proof}
($\Rightarrow$) Based on Definition \ref{def:k-safe(cyc)}, if $R$ is in $k$-$\mathsf{safe}({\Phi}_\Delta)$, then for all $k$-cycles $\sigma$ either $\Phi(R,\mathsf{Rule}(\sigma))=T$, or for all renaming functions $rn^{\ast}$ for $I^\sigma$, $\sigma$ is not active w.r.t. restricted critical database $rn^{\ast}(I^{\sigma})$. Therefore, Algorithm \ref{alg:kSafe} returns $T$.
\\
($\Leftarrow$)
By Proposition \ref{lem:arbitraryDB-renamedDB}, for all $k$-cycles $\sigma$ and for all renaming functions $rn^{\ast}$ for $I^\sigma$, if $\sigma$ is not active w.r.t. restricted critical database $rn^{\ast}(I^{\sigma})$, then the given rule set belongs to $k$-$\mathsf{safe}({\Phi}_\Delta)$, and by Theorem \ref{thm:kSAFE}, is all-instance terminating.
\end{proof}

\begin{thm} \label{thm:k-safe_comp}
Let $R$ be a given rule set and $k\ge 0$ be a unary-encoded integer.
Assuming that checking $\Delta$ can be done in \textsc{PTime},
the complexity of checking membership in $k$-$\mathsf{safe}(\Phi)$ is in $\textsc{PSpace}$.
\end{thm}

\begin{proof}
Given a rule set $R$ and an acyclicity condition $\Delta$, let us first guess a $k$-cycle $\sigma=(\sigma_1,\dots,\sigma_{(k+1)\times|R|-1})$ based on $R$, and then check whether $\mathsf{Rule}(\sigma)\notin \Delta$.
The guessing part can be done using a non-deterministic algorithm. Furthermore, based on our assumption, the checking part can be done in \textsc{PTime}.

For the guessed $k$-cycle $\sigma$, we then proceed by guessing a renaming function $rn^{\ast}$ and a restricted chase sequence $\mathcal{I} : rn^{\ast}(I^{\sigma}),\dots,I_{(k+1)\times|R|-1}$ constructed from $\sigma$ using a tuple of chained homomorphisms $H=(h_1,\dots,h_{(k+1)\times|R|-1})$, and verifying whether $\mathcal{I}$ is chained by checking whether $\sigma$ is active w.r.t. $rn^{\ast}(I^{\sigma})$, which gives us the complement of the desired membership checking problem.

An iterative procedure is required to construct $\mathcal{I}$.
In each step $i>0$ of this procedure we need to remember each instance $I_{i-1}$ in the constructed sequence, guess a homomorphism $h_i$, and proceed to derive $I_{(k+1)\times|R|-1}$. For this purpose, we need $\textsc{NSpace}(((k+1)\times|R|-1)\times \beta)$ memory space to remember intermediate instances, where $\beta$ is the maximum number of head atoms of rules in $\sigma$. In addition, guessing each homomorphism $h_i$ can be done using an \textsc{NP} algorithm and having access to an \textsc{NP}-oracle, verifying if $h_i$ can be extended a homomorphism $h'_i$ and leads to a chained tuple of homomorphisms is \textsc{NP}-complete \cite{rutenburg1986complexity}.
All these tasks can be maintained within the same $\textsc{NSpace}(((k+1)\times|R|-1)\times \beta)$ complexity bound, giving us a $co\textsc{NSpace}(((k+1)\times|R|-1)\times \beta)$ upper bound for the complexity of membership checking.

As a corollary to Savitch's theorem \cite{savitch1970relationships}, we have \textsc{PSpace}=\textsc{NPSPace}. 
Also, based on Immerman–Szelepcsényi theorem \cite{immerman1988nondeterministic}, non-deterministic space complexity classes are closed under complementation.
Therefore, based on the above analysis, the complexity upper bound for the membership checking problem is in $\textsc{PSpace}$.
\end{proof}

\begin{rem}\label{remark1}
	Based on Theorem \ref{thm:k-safe_comp}, it can be seen that for $\Delta\in \{\WA,\JA,\SWA\}$,\footnote{SWA denotes the {\em super-weak acyclicity} condition of skolem chase terminating rule sets \cite{marnette2009generalized}.} the complexity of checking \kSAFE$(\Phi_{\Delta})$ is in \textsc{PSpace}. This shows that our conditions, when considering skolem acyclicity criteria for which membership checking can be done in \textsc{PTime}, are easier to check than even the easiest known condition of the restricted chase in the literature (i.e., RJA) for which the complexity of membership checking is \textsc{ExpTime}-complete.
	
	In addition, for semantic conditions of terminating skolem chase, such as MSA (respectively, MFA), checking $\Delta$ cannot be done in \textsc{PTime} and a worst-case complexity of \textsc{ExpTime}-complete (respectively, \textsc{2ExpTime}-complete) can be computed \cite{grau2013acyclicity}. It follows that for the membership checking problem of \kSAFE$(\Phi_{\Delta})$, where $\Delta$ is MSA (respectively, MFA), an \textsc{ExpTime}-complete 	(respectively, \textsc{2ExpTime}-complete) complexity can be computed. The hardness proof can be established from the membership checking problem of the corresponding class with terminating skolem chase since this problem cannot be easier than that in general.
\end{rem}

\section{Extension of Bounded Rule Sets}
\label{bounded}

In \cite{zhang2015existential}, a family of existential rule languages with finite skolem chase based on the notion of {\em $\delta$-boundedness} is introduced and the data and combined complexities of reasoning with those languages for {\em $k$-exponentially bounded functions} are obtained.
Utilizing a parameter called {\em bound function}, our aim in this section is to show how to extend bounded rule sets from the skolem to restricted chase. In particular, we show that for any class $\Delta$ of terminating rule sets under the skolem chase, there exists a more general class of terminating rule sets under the restricted chase that extends $\Delta$. We show how to construct such an extension, and we analyze the membership and reasoning complexities for extended classes.
First, let us introduce some terminologies. 

A {\em bound function} is a function from positive integers to positive integers. A rule set $R$ is called {\em ${\delta}$-bounded under the skolem chase} for some bound function $\delta$, if for all databases $I$, $ht(${\small$\mathsf{chase}_{sk}(I,R)$}$)$ $\leq \delta(||R||)$, where $||R||$ is the number of symbols occurring in $R$. Given an instance $I$, $ht(I)$ denotes the height (maximum nesting depth) of terms that have at least one occurrence in $I$, if it exists, and $\infty$ otherwise. In this paper, when we mention $\delta$ as a bound function, we assume that $\delta$ is computable.

Let us denote by $\delta$-${\cal B}^{sk}$ the class of $\delta$-bounded rule sets under the skolem chase. For the restricted case, the definition is similar.

\begin{defn} \label{def:delta_bounded}
	Given a bound function $\delta$, a rule set $R$ is called {\em$\delta$-bounded under the restricted chase},\footnote{Note that by definition, the fairness condition is a requirement for a non-terminating restricted chase sequence.}
denoted $\delta$-${\cal B}^{res}$, if for all databases $I$ and for any restricted chase sequence ${\cal I}$ of $R$ and $I$, $ht({\cal I}) \leq \delta(||R||)$.
\end{defn}

\begin{exm} \label{exm:R1 bounded}
	For the rule set $R_1$ of Example \ref{exm:simplified}, it can be seen that the height of skolem terms in any restricted chase sequence is no more than $3$. Therefore, $R_1$ is $\delta$-bounded under the restricted chase variant for some bound function $\delta$ for which $\delta(||R_1||)=3$. It is worth noting that $R_1$ does not belong to $\delta$-bounded rule sets for any computable bound function $\delta$ under the skolem chase. 
\end{exm}

Before diving into more details, let us first demonstrate the relationship between $\delta$-bounded rule sets and $k$-$\mathsf{safe}({\Phi}_\Delta)$ rule sets as given in Proposition \ref{prop:bounded-ksafe} below.

\begin{prop}\label{prop:bounded-ksafe}
Let $R$ be a $k$-$\mathsf{safe}(\Phi)$ rule set in which $k$ is a unary encoded integer computable in $\mathcal{O}(P(n))$, for some function $P(n)$.
Then $R$ is $\delta$-bounded under the restricted chase for some function $\delta$ that is computable in $\mathcal{O}(P(2\times \log ||R||))$.\footnote{Here $n$ denotes the size of representation for the parameter of $k$. We say that $k$ can be computed in $\textsc{DTime}(P(n))$ if: There is a deterministic Turing machine $M$ such that, given an integer $l>0$, $M$ outputs $k(l)$ in $P(\log l)$ stages. Note that $\log l$ is the size of binary representation of $l$.}
\end{prop}

\begin{proof}
Let $R$ be $k$-$\mathsf{safe}(\Phi)$.
Based on Definition \ref{def:k-safe(cyc)}, for each $k$-cycle $\sigma$ which is mapped to $F$ under $\Phi$, $\sigma$ is safe (i.e., for all databases $I$, $\sigma$ is not active w.r.t.\,$I$).
Each rule application in a chained sequence can increase the depth of a skolem term at most by one. Henceforth, the longest possible chained sequence provides an upper bound for the term depth. We show this upper bound is $k\times (k+2)$.

This is because the length of the longest such sequence for a $k$-cycle is upper bounded by $k\times (k+1)$, and therefore, any sequence of length $k\times (k+2)$ must contain at least one $k$-cycle. Since no $k$-cycle is active w.r.t. any database, the depth of any skolem term generated by the longest chained sequence is less than $k\times (k+2)$. Thus $R$ is $k\times (k+2)$-bounded, which gives a quadratic bound in $k$.
Since $k$ is computable in $\mathcal{O}(P(n))$ and it is unary represented, then $k^2$ is computable in $\mathcal{O}(P(2\times \log ||R||))$, where $\log ||R||$ is the size of binary representation of $||R||$.
%When $||R||$ is the size of the input, the number of $k$-cycles is bounded by $\mathcal{O}(2^{||R||})$.
%This yields the following cost of computing $\Phi$: $\mathcal{O}(P(\log 2^{||R||}))=\mathcal{O}(P(||R||))$.
%%Assuming that the cost of computing $\Phi$ is $\mathcal{O}(P(n))$, for some function $P(n)$, the cost of computing 
%%$k$ (resp.%$\Phi$ is $\mathcal{O}(P(n))$%(resp. $\mathcal{O}(P_2(n))$),
Based on the above argument, we conclude that such a bound function always exists, and $\mathcal{O}(P(2\times \log ||R||))$ is an upper bound for the cost of computing the bound function.
\end{proof}
%Based on the above argument, assuming that the cost of computing $k$ (resp. ${\Phi}_{\Delta}$) is $\mathcal{O}(P_1(n))$ (resp. $\mathcal{O}(P_2(n))$), we conclude that such a bound function always exists, and the sum of $\mathcal{O}(P_1(n)^2)$ and $\mathcal{O}(P_2(n))$ is an upper bound for the cost of computing the bound function.

%Let us focus on computable bound functions $\delta$.
In what follows, we present our results on the membership of $\delta$-bounded rule sets under the restricted chase variant. 
Before we proceed, let us define what we mean by membership in the context of this chase version.
The problem of membership for the skolem chase is to check if all skolem chase sequences halt (terminate) before the maximum height of skolem terms in each sequence reaches $\delta(||R||)$ for all databases. 
As described in \cite{zhang2015existential}, checking membership for $\delta$-bounded rules under the skolem chase can be precisely characterized using only one chase sequence and utilizing the Marnette's critical database technique \cite{marnette2009generalized}, on a single database which is constructed from the given rule set only once.

On the other hand, one can not determine the membership in the $\delta$-bounded rules under the restricted chase using a single chase sequence. For this purpose, all possible restricted chase sequences need to be considered. Furthermore, restricted critical databases introduced in Definition \ref{def:criticalDB} can help us determine whether a possible chase sequence constructed from a given rule set witnesses the non-terminating status of the rule set under the restricted chase.

In what follows, we propose a procedure for membership checking of $\delta$-bounded rule sets under the restricted chase.
Given a rule set $R$ and a bound function $\delta$, the procedure $MembCheck(R,\delta)$ is defined as follows:
\begin{itemize}
\item 
Check whether $R$ is $\delta$-bounded under the skolem chase using the (skolem) critical database constructed from $R$, denoted $I^{R}$. If true, returns $T$.
\item 
Otherwise, for some $i > 0$, the height of 
$\mathsf{chase}_{sk}^{i}(I^{R},R)$ is $\delta(||R||)+1$; for each skolem chase sequence generated by a path $\pi = (r_1,\dots, r_n)$ that reaches the height of $\delta(||R||)+1$, 
we check whether $\pi$ is active w.r.t. the restricted critical database $rn^{\ast}(I^{\pi})$ for all renaming functions $rn^{\ast}$. If the answer is no for all such paths $\pi$, then the procedure returns $T$, otherwise it returns $F$ (false).
\end{itemize}

A $T$ answer means that $R$ is $\delta$-bounded under the restricted chase and an $F$ answer means that it is unknown whether $R$ is $\delta$-bounded under the restricted chase or not. The reason for the latter case is that when the skolem chase reaches the height of $\delta(||R||)+1$ by a path $\pi = (r_1, \dots, r_n)$, although we can check activeness of $\pi$ w.r.t. restricted critical databases, we may not be able to determine whether such a path leads to at least one fair sequence.

\begin{prop}\label{prop:MembCheck-sound}
Given a bound function $\delta$ and an arbitrary rule set $R$, $MembCheck(R,\delta)$ is sound, i.e., if it returns $T$, then $R$ is $\delta$-bounded under the restricted chase. Furthermore, if $R$ consists of rules with single-head, then $MembCheck(R,\delta)$ is sound and complete.
\end{prop}

Note that the completeness problem is as follows: $MembCheck(R,\delta)$ is complete if for any given rule set $R$ and bound function $\delta$, if $MembCheck(R,\delta)=F$, then $R$ is not $\delta$-bounded under the restricted chase. 
%, $MembCheck(R,\delta)$ returns $T$.

\begin{proof} %(Proof of Soundness)
Let $\delta$ be a bound function. By \cite{marnette2009generalized}, it suffices to use the skolem critical database $I^R$ to capture all skolem chase sequences w.r.t. any database $I$, so that 
$ht(${\small$\mathsf{chase}_{sk}(I,R)$}$)$ $\leq \delta(||R||)$ only if $ht(${\small$\mathsf{chase}_{sk}(I^R,R)$}$)$ $\leq \delta(||R||)$.
Consequently, if $R$ is $\delta$-bounded under the skolem chase w.r.t.\,$I^R$, it is $\delta$-bounded under the skolem chase w.r.t any database $I$, and by the relationship between the skolem and restricted chase, $R$ is $\delta$-bounded under the restricted chase w.r.t any database $I$.

Otherwise, for each path $\pi$ that leads to some skolem chase sequence that reaches the height of $\delta(||R||)+1$, $\pi$ being not active 
w.r.t.\,$rn^{\ast}(I^{\pi})$ for all renaming function $rn^{\ast}$ for $I^\pi$ implies, by Theorem \ref{key0}, that $\pi$ is not active w.r.t. any database. When all chained sequences of path $\pi$ fail to reach the height of $\delta(||R||)+1$, no restricted chase sequence of $\pi$ can reach that height because an unchained sequence does not expand skolem terms cumulatively throughout.
It follows that the largest height by any database is bounded by $\delta(||R||)$. This gives the desired conclusion for the soundness of $MembCheck$ for arbitrary rules.\footnote{If there exists such a path $\pi$ that is active and leads to a restricted chase sequence, which by default must be fair, then we can decide that $R$ is not $\delta$-bounded under the restricted chase (again, the fairness condition must be satisfied). In this case, the procedure is complete by returning $F$.
On the other hand, 
if all such paths $\pi$ lead {\em only} to unfair restricted chase sequences (i.e., infinite chase sequences generated by active triggers in Definition \ref{restricted} without requiring the fairness condition), then no restricted chase sequence has reached beyond the bound and in this case, that our procedure returns $F$ shows its incompleteness. But in general, the problem of whether such a $\pi$ leads {\em only} to unfair chase sequences may be undecidable.}

%To prove the completeness of single-head rule sets, %for each such rule set $Rr
For any single-head rule set $R$, from \cite{gogacz2019all}, we know that the fairness condition can be safely neglected, i.e., the existence of a (possibly unfair) infinite restricted chase sequence implies the existence of a fair one.
Therefore, $R$ is not $\delta$-bounded.
\end{proof}

\begin{prop} \label{prop:membership}
	Let $R$ be a rule set and $\delta$ a bound function computable in $\textsc{DTime}(P(n))$\footnote{The class of complexity languages decidable in time $P(n)$ using a deterministic Turing machine. \textsc{NTime} is defined similarly but using a non-deterministic Turing machine.} for some function $P(n)$. 
	Then, it is in
	$$co\textsc{NTime}(\mathsf{C}_{\delta} 
	+
	||R||^{{||R||}^{\mathcal{O}(\delta(||R||))}}
	)$$
	to check if $MembCheck(R,\delta)$ returns $T$, where $\mathsf{C}_{\delta} = P(\log ||R||)^{\mathcal{O}(1)}$.
\end{prop}

\begin{proof}
	For the skolem chase with skolem critical database, from Proposition 6 of \cite{zhang2015existential} we know that using the critical database technique of \cite{marnette2009generalized},
 the maximum number of atoms generated in a skolem chase sequence is bounded by $||R||^{||R||^{\mathcal{O}(\delta(||R||))}}$, which is also an upper bound for the number of atoms generated in a restricted chase sequence.  
	
From \cite{marnette2009generalized} we know that in the case of the skolem chase if any sequence terminates on a rule set $R$ and a database $I$, then the instances returned by all sequences are isomorphically equivalent. So, for $\delta$-boundedness for the skolem chase, it suffices to consider only one sequence. But for the case of the restricted chase, we need to consider all such sequences.

Given a rule set $R$ and a bound function $\delta$, the procedure $MembCheck(R,\delta)$ first checks whether $R$ is $\delta$-bounded under the skolem chase.

For the complexity of this check, we need to consider the size of each skolem chase sequence to produce the height of $\mathcal{O}(\delta)$ that is upper bounded by $||R||^{||R||^{\mathcal{O}(\delta(||R||))}}$, which can be computed in $\textsc{DTime}(||R||^{||R||^{\mathcal{O}(\delta(||R||))}})$.
In addition, an upper bound for the chase of size $||R||^{||R||^{\mathcal{O}(\delta(||R||))}}$ can be computed in $\textsc{DTime}\big((||R||+ P(\log ||R||))^{\mathcal{O}(1)}\big)$.	
Therefore, according to \cite{zhang2015existential}, the overall complexity of this check is: $\textsc{DTime}\big((P(\log ||R||))^{\mathcal{O}(1)}+ ||R||^{||R||^{\mathcal{O}(\delta(||R||))}}\big)$.
	
If the above condition is not satisfied (i.e., some $R$ is not $\delta$-bounded under the skolem chase), for some $i\ge 1$, the height of {\small$\mathsf{chase}_{sk}(I^{R},R)$} is $\delta(||R||)+1$. So, for each skolem chase sequence that is generated by a path $\pi=(r_1,\dots,r_n)$ which reaches the height of $\delta(||R||)+1$, for all renaming functions $rn^{\ast}$ for $I^\pi$, we check whether $\pi$ is active w.r.t. $rn^{\ast}(I^{\pi})$. A {\em no} answer to the above check yields a $T$ output from $MembCheck(R,\delta)$.

Based on the above argument, to proceed, using a non-deterministic algorithm we first guess a sequence of triggers $\bigcup_{i=1}^{N}(r_i,h_i)$, where $N$ is upper bounded by $||R||^{||R||^{\mathcal{O}(\delta(||R||)+1)}}=||R||^{||R||^{\mathcal{O}(\delta(||R||))}}$ that can lead to the construction of a skolem chase sequence $\mathcal{I}$, and a renaming function $rn^{\ast}$.

Then we need to verify if $\mathcal{I}$ is active w.r.t. $rn^{\ast}(I^{\pi})$, where $\pi$ is the path constructed from the guessed $r_i$'s. For the latter, for each projection $\pi'$ of $\pi$, first, to verify the chained property, we determine if each rule in $\pi'$ depends on some previous rule in the path. The complexity of this latter verification task is quadratic in the size of the guessed chase sequence. 

Furthermore, given path $\pi$, the maximum number of chained restricted chase sequences is bounded by 
${||R||^{\mathcal{O}(\delta(||R||))}}$, and since the length of the guessed sequence is bounded by $\mathcal{O}(||R||^{||R||^{\mathcal{O}(\delta(||R||))}})$, verifying if $\mathcal{I}$ is active w.r.t. $rn^{\ast}(I^{\pi})$ is at most polynomial in $||R||^{||R||^{\mathcal{O}(\delta(||R||))}}$ which can be implemented in $\textsc{NTime}(||R||^{||R||^{\mathcal{O}(\delta(||R||))}})$. Similar to the proof of Theorem \ref{thm:k-safe_comp}, the construction of renaming functions can take at most polynomial in the size of $\pi$ which can be done in $\textsc{NTime}(||R||^{||R||^{\mathcal{O}(\delta(||R||))}})$. So, clearly, all the above tasks can be maintained in $\textsc{NTime}\big((P(\log ||R||))^{\mathcal{O}(1)}+ ||R||^{||R||^{\mathcal{O}(\delta(||R||))}}\big)$.

The membership is complement to the above problem, and therefore, belongs to $co\textsc{NTime}(\mathsf{C}_{\delta} 
	+
	||R||^{{||R||}^{\mathcal{O}(\delta(||R||)}}
	)$ as desired.
\end{proof}

Next, we investigate membership and reasoning complexities of bounded rule sets under what is called {\em exponential tower functions}, which are defined as follows:

\begin{center}
	$\mathsf{exp}_{\kappa}(n)= \begin{cases} 
	n & \kappa= 0 \\
	2^{\mathsf{exp}_{\kappa-1}(n)} & \kappa> 0
	\end{cases}$
\end{center}

Since the complexity of checking $\delta$-bounded property of Proposition \ref{prop:membership} is dominated by the second term inside $co\textsc{NTime}$, if $\delta(n)=\mathsf{exp}_{\kappa}(n)$, then its overall complexity increases by two exponentials.
We thus have

\begin{cor}\label{cor:expk_bounded}
	Given a rule set $R$ checking if
	$MembCheck(R,\mathsf{exp}_{\kappa})$ returns $T$ is in $co\textsc{N}(\kappa+2)$-$\textsc{ExpTime}$.
\end{cor}

\begin{exm}
	Based on the observation made in Example \ref{exm:R1 bounded}, the rule set $R_1$ in Example \ref{exm:simplified} is $\mathsf{exp}_{0}$-bounded under the restricted chase, however it does not belong to $\mathsf{exp}_{\kappa}$-bounded ontologies under the skolem chase for any computable $\kappa$.
\end{exm}

\medskip

\noindent{\bf Data and Combined Complexity:} 

Now, let us investigate the reasoning complexities. The problem under consideration is Boolean Conjunctive Query (BCQ) answering which is defined as follows.
Given rule set $R$, a database $I$ and a Boolean query $q$, decide if $I\cup R\models q$. The complexity of this problem is also known as {\em combined complexity} since the input size is the combined size of all $I$, $R$, and $q$. In the BCQ answering problem if $R$ and $q$ are fixed and only $I$ changes, then it is called {\em data complexity}.
Focusing on $\mathsf{exp}_{\kappa}$-bounded rule sets under the restricted chase variant, we have the following results on reasoning complexities.

\begin{thm} \label{thm:combined}
	The problem of Boolean conjunctive answering for $\mathsf{exp}_{\kappa}$-bounded rule sets under the restricted chase variant is in $(\kappa+2)$-$\textsc{ExpTime}$-complete for combined complexity and $\textsc{PTime}$-complete for data complexity.
\end{thm}

\begin{proof}
	Let $R$ be an $\mathsf{exp}_{\kappa}$-bounded rule set under the restricted chase variant and $I$ be a database.
	Then, let us guess a restricted chase sequence $\mathcal{I}$, non-deterministically. With an argument similar to that of the proof of Proposition \ref{prop:membership} in which $\delta(n)=\mathsf{exp}_{\kappa}(n)$, we know that the number of atoms of $\mathcal{I}$ is bounded by $||R||^{||R||^{\mathsf{exp}_{\kappa}(||R||)}}=\mathcal{O}(\mathsf{exp}_{\kappa+2}(||R||))$.

	Membership follows since the entailment of a BCQ $q$ can be shown by finding such a sequence $\mathcal{I}: I=I_0,\dots,I_n$ based on $R$ such that $I_n$ satisfies $q$ according to the following fact from \cite{fagin2003data}:
	Let $J$ and $K$ be two finite instances returned by the restricted chase of an $\mathsf{exp}_{\kappa}$-bounded rule set $R$ and a database $I$. Then $K$ and $J$ are homomorphically equivalent. Based on the above fact and the homomorphic equivalence classes, in the rest of this proof, we let $chase_{res}(I,R)$ denote one representative of the equivalence class for all results of the restricted chase of $R$ and $I$. In addition, based on \cite{fagin2003data}, it is known that $chase_{res}(I,R)\models R$, and also there is a homomorphism from $I$ to $chase_{res}(I,R)$. Furthermore, $I\cup R\models q$ if and only if $chase_{res}(I,R)\models q$.
	
Let $k$ and $n$ denote the number of relation symbols and the maximal arity of relation symbols appearing in $R$, respectively. 
Let further $l$ and $m$ represent the number of function symbols,
and the maximal arity of function symbols appearing in $sk(R)$, respectively.
In addition, let $c$ denote the number of constants appearing in $I$, and $Q(\mathbf{t})$ be a fact in $chase_{res}(I,R)$.
It is easy to verify that the number of symbols in each constituent $t\in \mathbf{t}$ is upper bounded by $\sum_{i=0}^{\mathsf{exp}_{\kappa}(||R||)}m^i=m^{\mathcal{O}(\mathsf{exp}_{\kappa}(||R||))}$. Also, it is clear that each symbol is either a constant or a function symbol. Therefore, the number of facts in $chase_{res}(I,R)$ is upper bounded by $(c+l)^{m^{\mathcal{O}(\mathsf{exp}_{\kappa}(||R||))}\times n}\times k$. Since $k, n, l , m\le ||R||$, and $c=|dom(I)|$, the following upper bound is derived for the number of facts in $chase_{res}(I,R)$:
$(|dom(I)|+||R||)^{||R||^{\mathcal{O}(\mathsf{exp}_{\kappa}(||R||))}\times ||R||^{\mathcal{O}(1)}}$, which can be computed in $\textsc{DTime}\big((|dom(I)|+||R||)^{||R||^{\mathcal{O}(\mathsf{exp}_{\kappa}(||R||))}}\big)$.

To compute the reasoning complexity involving a BCQ $q$, it is now sufficient to evaluate $q$ on $chase_{res}(I,R)$ directly.\footnote{Without loss of generality, we assume that $q$ is in {\em prenex normal form}.} To continue the analysis, we only need the number of existential variables occurring in $q$, which we denote by $v$. Then we need to check whether there is a substitution $h$ which maps every existential variable in $q$ to a ground term of height less than $\mathsf{exp}_{\kappa}(||R||)$, such that $h(q) \subseteq chase_{res}(I,R)$.
From the previous analysis is clear that $(|dom(I)|+||R||)^{{||R||^{\mathcal{O}(\mathsf{exp}_{\kappa}(||R||))}}\times v}$ substitutions need to be checked. Since $v\le ||q||$, the evaluation of checking whether $h(q)\subseteq chase_{res}(I,R)$ can be done in $$\textsc{DTime}\big(
	(|dom(I)|+||R||)^{{||R||}^{(\mathsf{exp}_{\kappa}(||R||)}\times ||q||^{\mathcal{O}(1)}}\big)$$
Hence, a $(\kappa+2)$-$\textsc{ExpTime}$ upper bound can be computed for the combined complexity, as desired.

We can use a construction similar to that of \cite{zhang2015existential} for the hardness proof.
We briefly sketch it here. Let us consider a deterministic Turing machine $M$ which terminates in $\mathsf{exp}_{\kappa+2}(n)$ number of steps on any input of length $n$. Let us assume that the query and data schema is a singleton set $\{\mathsf{Accept}\}$ and $\emptyset$, respectively, where $\mathsf{Accept}$ is a nullary relation symbol. We need to show that for each input $x$ that is a binary string of length $n$, there is an $\mathsf{exp}_{\kappa}$-bounded rule set under the restricted chase variant such that $M$ terminates on $x$ if and only if $\emptyset \cup R\models \mathsf{Accept}$. To construct the rule set $R$ we need to define a linear order of length $\mathsf{exp}_{\kappa+2}(n)$ on integers which are represented in binary strings from $0$ to $\mathsf{exp}_{\kappa+2}(n)$. Once a linear order is defined, we can construct a set of existential rules to encode the Turing machine $M$ and the input $x$. 
Once we have such a construction, we can establish the lower bound on the combined complexity of reasoning with existential rules under the restricted chase.
This lower bound combined with the upper bound derived above provides the exact bound for the combined complexity of $\mathsf{exp}_{\kappa}$-bounded rule sets under the restricted chase.

Furthermore, the data complexity of query answering with $\mathsf{exp}_{\kappa}$-bounded rule sets under restricted chase is \textsc{PTime}-complete.
The \textsc{PTime} upper bound for the data complexity can be derived from the above analysis, and the hardness follows from the $\textsc{PTime}$-completeness of data complexity of Datalog, cf. \cite{dantsin2001complexity}.
\end{proof}

\section{Experimentation}\label{experiments}

To evaluate the performance of our proposed methods for termination analysis, we implemented our algorithms in Java on top of the {\em \textsc{Graal} rule engine} \cite{baget2015graal}. 
Our goal was twofold: 1) to understand the relevance of our theoretical approach with real-world applications, and 2) to understand the computational feasibility \-- even though the problem of checking semantic acyclicity conditions, such as checking activeness of all $k$-cycles w.r.t. restricted critical databases have a high theoretical worst-case complexity, it may still be a valuable addition to the tools of termination analysis in real-world scenarios.

We looked into a random collection of 700 ontologies from The Manchester OWL Corpus (MOWLCorp) \cite{matentzoglu_2014_10851}, which is a large corpus of ontologies on the web.
%, that at the time of conducting our experiments included $20,996$ ontologies. {\color{red} ???  two numbers ?}
This corpus is a recent gathering of ontologies through sophisticated web crawls and filtration techniques.
After standard transformation into rules (see \cite{cuenca2013acyclicity} for details),\footnote{Due to limitations of this transformation, our collection does not include ontologies with nominals, number restrictions or denial constraints.} based on the number of existential variables occurring in transformed ontologies, we picked ontologies from two categories of up to $5$ and $5$-$200$ existential variables with equal probability ($350$ from each).
We ran all tests on a Macintosh laptop with 1.7 GHz Intel Core i7 processor, 8GB of RAM, and a 512GB SSD, running macOS Catalina.

\subsection{Implementation Setup}

Here, we provide the details on our implementation to identify $k$-$\mathsf{safe}(\Phi_{\Delta})$ rule sets. 

For a given $k\geq 0$ and a class $\Delta$ (which also denotes the corresponding acyclicity condition) of finite skolem chase, to start, the {\em candidate pool} of ontologies which is considered for $k$-$\mathsf{safe}(\Phi_{\Delta})$ is the collection of all ontologies. The ontologies that fail our tests for $k$-$\mathsf{safe}(\Phi_{\Delta})$ will be removed. Then at the end of this process, we obtain a set of terminating ontologies.

For each given ontology, we transform it to a rule set $R$. In our experiments, we consider extending four classes of finite skolem chase, $\Psi = \{\WA, \JA, \aGRD, \text{MFA}\}$. For each $k$-cycle $\sigma$ based on $R$, first by using the technique of piece-unification, we may eliminate $R$ from the candidate pool. If not removed, we then check whether ${\mathsf{Rule}}(\sigma)$ satisfies the acyclicity condition $\Delta\in \Psi$. If not, we run experiments to check whether $\sigma$ is active w.r.t. its restricted critical databases. 

Let us first introduce the technique based on piece-unification.

	\begin{defn} ({\bf Piece-unification} \cite{baget2009extending})\label{defn:piece_unifier}
		Given a pair of rules $(r_1,r_2)$, a {\em piece-unifier} of $body(r_2)$ and $head(r_1)$ is a unifying substitution $\theta$ of $var(B)\cup var(H)
		$ where $B\subseteq body(r_2)$ and $H\subseteq head(r_1)$ which satisfies the following conditions: 
		\begin{itemize}
		\item 
		[(a)] 
		 $\theta(B)=\theta(H)$, and 
		\item [(b)] 
		variables in $var_{ex}(H)$ are unified only with those occurring in $B$ but not in $body(r_2)\setminus B$.
		\end{itemize}
	\end{defn}

	Condition (a) gives a sufficient condition for rule dependency, but it may be an overestimate, which is constrained by condition (b).
	Note that in Example \ref{exm:counter}, condition (a) holds for $B=\{T(x,y)\}$ and $H=\{T(y,z)\}$ where $\theta=\{x/y,y/z\}$, and condition (b) does not, since $var_{ex}(H)=\{z\}$ and $z$ unifies with $y$ which occurs in both $B$ and $body(r)\setminus B = \{P(x,y)\}$.
	Therefore, no piece-unifier of $body(r)$ and $head(r)$ exists.

	Piece-unification is known to provide a necessary condition for rule dependencies in that 
	for any two rules $r$ and $r'$, if $body(r)$ and $head(r')$ are not piece-unifiable, then no trigger $(r,h)$ exists that relies on some atom derived  
	from $head(r')$ (cf. Property 18 of \cite{baget2011rules}).
	Below, given a substitution $\theta$, $dom(\theta)$ denotes the domain of $\theta$, which is the set of substituted variables in $\theta$, and 					
	$codom(\theta)$ denotes the co-domain of $\theta$, which is the set of substitutes in $\theta$. For technical reasons, if $\theta$ is a piece-unifier of 
	$body(r)$ and $head(r')$, then $dom(\theta)$ refers to the subset of substituted variables which also appear in $body(r)$ and $codom(\theta)$ 
	refers to the subset of substitutes which appear in $body(r)$ as well.

If the set of all sequences of piece-unifiers that can be constructed from a path $\pi$ is non-empty, then for each sequence of piece-unifiers that can be formed in $\pi$, we need to check whether this sequence leads to a restricted chase sequence or not.

To show whether each sequence of piece-unifiers leads to a sequence of rules which are transitively-dependent, checking if they only satisfy conditions (a) and (b) above is not sufficient.
Indeed, as shown in \cite{baget2011rules}, given two rules $r_1$ and $r_2$, $r_2$ depends on $r_1$ if and only if there is a piece-unifier $\theta$ of $body(r_2)$ with $head(r_1)$ such that $\theta$ satisfies the following conditions: (i) {\em atom-erasing}, and (ii) {\em productive} (a.k.a. {\em useful}, cf. \cite{baget2014revisiting}).
The former condition checks that $\theta(body(r_2))$ is not included in $\theta(body(r_1))$.
In addition, the productivity condition for $\theta$ means that $\theta(head(r_2))$ is not included in $\theta(body(r_1))\cup \theta(head(r_1))\cup \theta(body(r_2))$. Note that the above two conditions can naturally be extended to sequences of piece-unifiers.
Therefore, in order to show that each path $\pi$ does not lead to a chained sequence, it suffices to show that each sequence of piece-unifiers constructed from $\pi$ (if any), does not satisfy either atom-erasing or productive condition.

The goal of this part is to present how we can eliminate the {\em irrelevant} $k$-cycles in our analysis.
For this purpose, we utilize the notion of piece-unification as follows, for a given $k >0$.
\begin{itemize}
\item
For each $k$-cycle $\sigma$, if the set of sequences of piece-unifiers is $\emptyset$, then $\sigma$ will be removed from consideration of further checks since $\sigma$ trivially leads to a terminating skolem chase before all the rules in $\sigma$ are applied (and therefore, a terminating restricted chase).
\item
For each $k$-cycle $\sigma$, if none of the sequences of piece-unifiers that can be constructed from $\sigma$ satisfy both conditions of atom-erasing and productive, then $\sigma$ is removed from our analysis.
\end{itemize}

We call each $k$-cycle which has not been removed during the abovementioned steps, {\em relevant}.\footnote{Note that the notion of {\em compatible unifiers} is introduced in \cite{baget2014extending} in which piece-unification has been relaxed to take into account arbitrary long sequences of rule applications.
This is similar to our goal. In fact, compatible unifiers provide a tighter notion which can help in removing more irrelevant $k$-cycles.}

In our experiments, we performed the following steps:
\begin{enumerate}
\item 
Transforming ontologies in the considered corpus into the normal form using standard normalization techniques (cf. \cite{carral2014mathcal}). This will ensure that concepts do not occur nested in other concepts and also each functional symbol introduced during normalization depends on as few variables in the rule as possible. It takes an input ontology path that can be parsed by the OWL API (which is in OWL/XML, OWL Functional Syntax, OBO, RDF/RDFS or Turtle format) (cf. \cite{horridge2011owl}) and produces a normalized ontology;
Note that we filter out the following axioms of input ontology: those that are not logical axioms and those containing datatypes, datatype properties, or built-in atoms as the conventional normalization methods are unable to handle them;
\item
Rewriting axioms to get first-order logic rules and writing them in the {\em dlgp format} (for ``Datalog+" \cite{baget2015datalog+});
\item 
Forming all relevant $k$-cycles $\Sigma$ constructed from each transformed rule set and for each $\sigma \in \Sigma$, where $\sigma=(r_1,\dots,r_n)$, we check if $\mathsf{Rule}(\sigma)\in \Delta$, for each $\Delta$ from \{WA, JA, aGRD, MFA\};
	\item
	 For each $\Delta$ from \{WA, JA, aGRD, MFA\} and for each relevant $k$-cycle $\sigma$ such that $\mathsf{Rule}(\sigma)\notin \Delta$,  we check the activeness of $\sigma$ w.r.t.\,$I^{\sigma}$, i.e., we check if there exists a chained tuple of homomorphisms $H=(h_1,\dots,h_n)$ for $\sigma$ at each step ($1 \leq i \leq n$). We implemented a chained homomorphism checker to accomplish this task;
\begin{itemize}
\item
During the above check, whenever a relevant $k$-cycle $\sigma$ is determined to be active w.r.t.\,$I^\sigma$, the rule set $R$ is removed from the candidate pool;
\item
If every $k$-cycle $\sigma$ is not active w.r.t.\,$I^\sigma$, we check the reason for the failure, say for rule $r_i~(1\leq i \leq n)$. If the failure is due to lack of a trigger which is caused by mapping multiple occurrences of a body variable of $r_i$ to distinct indexed constants, then we know, by Theorem \ref{key0}, that for some minimal renaming function $rn$ for $I^\sigma$, a trigger exists so that there is a chained restricted chase sequence from 
$rn(I^\sigma)$ up to (and including) $r_i$. However, we examined all the cases of failure and did not find any failure was caused this way. Therefore, there is no need to continue experiments using the updated restricted critical database as laid out in Theorem \ref{key0}. This is to say that the phenomenon illustrated in Example \ref{exm:counter3} did not show up in our collection of practical ontologies. 
\end{itemize}
\item
Ontologies in the remaining candidate pool are decided to be terminating.
	\end{enumerate}

\subsection{Experimental Results}
\label{sub}
For each ontology, we allowed 2.5 hours to complete all of these tasks. In case of running out of time or memory, we report no terminating result.
For the first experiment, we considered $k$-$\mathsf{safe}(\Phi_{\Delta})$ rule sets for four different cycle functions $\Phi_{\Delta}$ based on WA, JA, aGRD and MFA conditions, respectively, for different values of $k$.

	We consider WA since its acyclicity condition is the easiest to check. We consider three popular syntactic acyclicity conditions WA, JA, and aGRD because the main cost of checking $k$-$\mathsf{safe}(\Phi_{\Delta})$ is then on the extension provided in this paper. 
	Additionally, we consider MFA, a well-known semantic condition for checking the skolem chase termination, which is based on forbidding cyclic functional terms in the chase.
%: a rule set $R$ is MFA if the skolem chase of $I^{R}$ and $R$ does not have any cyclic functional term, where $I^{R}$ denotes the skolem critical database of $R$. {\color{red} (Definition???)}
Note that all other (syntactic) conditions considered in this paper are subsets of MFA. 
Besides, it is known that $\WA\subset \JA$ and aGRD is not comparable to either WA or JA.
We are interested to know whether the high worst-case complexity of our extension prohibits applications in the real-world.\footnote{For both MFA and RMFA, the complexity of membership checking is already higher than that of Algorithm 1 (assuming checking $\Delta$ is in \textsc{PTime}, cf. Remark \ref{remark1} 
in Section \ref{k-safe}).}

In Table \ref{table:1}, the results of these experiments are summarized where the values of columns 2-5 denote numbers of ontologies with properties provided in their first row.

\begin{table} []
	\caption{Membership among 700 ontologies in the collected corpora}\label{newcommands}
	\programmath
	\begin{tabular} {| c | c | c | c | c |}
	\cline{1-5}
		k    &   $k$-$\mathsf{safe}(\Phi_{aGRD})$    &    $k$-$\mathsf{safe}(\Phi_{\WA})$    &    $k$-$\mathsf{safe}(\Phi_{\JA})$    &   $k$-$\mathsf{safe}(\Phi_\text{MFA})$
		\\
		\cline{1-5}
 		 $k=0$  &
		163 &
		248 &
		299 &
		483
		 \\ 
		\cline{1-5}
		$k=1$   &
		171 &
		258 &
		310 &
		495
		\\
		\cline{1-5}
		$k=2$   &
		177 &
		264 &
		316 &
		501 
	    \\
		\cline{1-5}
		$k=3$   &
		182 &
		269 &
		321 &
		506
		\\
		\cline{1-5}
		$k=4$   &
		187 &
		274 &
		326 &
		511
		\\
		\cline{1-5}
		$k=5$   & 
		190 &
		277 &
		329 &
		514
		 \\
    	\cline{1-5}
		$k=6$  &
		192 &
		279 &
		331 &
		516
		\\
		\cline{1-5}
	\end{tabular}
	\unprogrammath
	\label{table:1}
\end{table}

\begin{table} [!hbt]
	\programmath
	\begin{tabular}{| c c c c |}
		\cline{1-4}
		\multicolumn{4}{| c |}
		{\multirow{2}{*}
		{\textbf{Average time analysis for $k=6$}}}\\
		& & & 		
		\\
		\cline{1-4}
		Classes  &  Avg. time (s)  &  T.W.A.T. (\#)  &  Terminating (\%)
		\\
		\cline{1-4}
		$6$-$\mathsf{safe}(\Phi_\text{aGRD})$  &  4139  &  125   &   27.4 \\
		\cline{1-4}
		$6$-$\mathsf{safe}(\Phi_\text{WA})$  &  3556  &  164  &  39.8 \\
		\cline{1-4}
		$6$-$\mathsf{safe}(\Phi_\text{JA})$  &  3231  &  183  &  47.2 \\
		\cline{1-4}
		$6$-$\mathsf{safe}(\Phi_\text{MFA})$  &  4923  &  282  &  73.7 \\
		\cline{1-4}
	\end{tabular}
	\unprogrammath
	\caption{Average time analysis for membership testing of terminating ontologies}
	\label{table:2}
\end{table}

Consider the case $k=0$. This is the case where we identify rule sets that are skolem chase terminating under three acyclicity conditions aGRD, WA, and JA as well as under the MFA condition. First, it is not surprising to observe that among 700 ontologies, the first three syntactic conditions identify only a small subset of terminating ontologies. However, when considering the MFA condition, we are able to capture many more rule sets as terminating in this collection.
Second, for our collection of practical ontologies, the gap between the terminating classes under aGRD and WA conditions is indeed non-trivial.
Interestingly, this appears to be the first time that these three syntactic classes of terminating rule sets are compared for practical ontologies. This shows that the theoretical advance from aGRD to WA may have significant practical implications.

As can be seen in Table \ref{table:1}, in all of the considered classes, by increasing $k$, the number of terminating ontologies increases. This is consistent with Theorem \ref{thm:kSAFE}. 
Our experiments stopped at $k=6$ as we did not find more terminating rule sets by testing $k=7$.

%{\color{blue}
We considered some optimizations in our implementation.
Before proceeding further, let us define some notions.
Given a rule set $R$, consider the graph of rule dependencies $\mathcal{G}_R$ of $R$ in which the set of nodes is $R$, and there is an edge from some node $r_i$ to a node $r_i$ if $r_j$ depends on $r_i$.
If there is a path from some rule $r_i$ to a rule $r_j$, then $r_i$ is called to be {\em reachable} from $r_j$.
If each node in $\mathcal{G}_R$ is reachable from each other node, then $\mathcal{G}_R$ is {\em connected}.
A component of $\mathcal{G}_R$ is a {\em maximal connected subgraph} of $\mathcal{G}_R$ (i.e., a connected subgraph of $\mathcal{G}_R$ with node set $X$ for which no larger set $Y$ containing $X$ is connected).

For each rule set $R$, we find the maximal connected subgraphs of $\mathcal{G}_R$ defined as above. Given an acyclicity condition $\Delta$, for each maximal connected subgraph $\mathcal{S}$ of $\mathcal{G}_R$, we check whether $\mathcal{S}\in \Delta$ returns true. If that is the case, then we do not need to check any path based on any non-empty subset of $\mathcal{S}$ for activeness.
The reason is that if $\mathcal{S}\in \Delta$, then any subset $\mathcal{S'}$ of $\mathcal{S}$ also satisfies $\Delta$. Therefore, any cycle $\sigma$ based on $\mathcal{S'}$ is safe.
This helped us remove irrelevant subsets of rules in 83 (11.8\%) of ontologies in our collection.

Given an acyclicity condition $\Delta$, by $\Phi_\Delta$ let us denote the cycle function constructed from $\Delta$. Then a notable subclass of $1$-$\mathsf{safe}(\Phi_\Delta)$ rules is called $\Delta^{\prec}$ introduced in \cite{grau2013acyclicity} which is defined as the set of rules $R$ in which each simple cycle in the graph of rule dependencies $\mathcal{G}_R$ of $R$ belongs to $\Delta$. In our collection, it can be seen that only one ontology is $\WA^{\prec}$ (and therefore, $\JA^{\prec}$), which belongs to $1$-$\mathsf{safe}(\Phi_\WA)$ but not in $1$-$\mathsf{safe}(\Phi_\aGRD)$ in Table \ref{table:1}. 
%{\color{red} (Where can this be seen?)}

%Additionally, it can be implied that if $\Delta$ is a position-based class of acyclicity, i.e., one that is based on acyclicity of some graph constructed from the dependency of rule positions, e.g., WA, JA or SWA, then $\Delta^{\prec}$ extends $\Delta$.
%Otherwise, if $\Delta$ is based on acyclicity of some graph constructed from rule dependencies, e.g., aGRD, or is an extension of such class, e.g., MFA, then $\Delta=\Delta^{\prec}$, i.e., $\Delta^{\prec}$ does not extend $\Delta$. {\color{red} (I don't understand the relevance.)}

%{\color{blue}
When $k$ grows from $0$, an interesting observation is that for each pair of acyclicity conditions $\Delta_1$ and $\Delta_2$ such that $\Delta_1\subset \Delta_2$, the rate of increase in the number of terminating rules under $\Phi_{\Delta_2}$ is faster than that of terminating rules under $\Phi_{\Delta_1}$ when $k$ grows from $k=0$ to $k=1$.
Then, for all $k>1$ the increase of terminating rules from $(k-1)$-$\mathsf{safe}(\Phi_{\Delta_i})$ to $k$-$\mathsf{safe}(\Phi_{\Delta_i})$ returns the same result for the case when either $i=1$ or $i=2$.

%{\color{red}The reason is that when $k=1$, in general more $k$-cycles need to be evaluated for activeness when checking membership in $k$-$\mathsf{safe}(\Phi_{\Delta_1})$ compared to the case of $k$-$\mathsf{safe}(\Phi_{\Delta_2})$.
%Also, recall that given some cycle function $\Phi_\Delta$, when $k$ is an integer, a rule set $R$ is in $k$-$\mathsf{safe}(\Phi_\Delta)$ if for each $k$-cycle $\sigma$ based on $R$, either (1) according to Definition \ref{def:cycle_function_acyclicity}, $\mathsf{Rule}(\sigma)$ satisfies the acyclicity condition under consideration ($\Delta$), or	(2) based on Definition \ref{def:k-safe(cyc)} and Proposition \ref{lem:arbitraryDB-renamedDB}, $\sigma$ is not active w.r.t. $rn^{\ast}(I^{\sigma})$.
%
%Therefore, given a rule set $R$, it may be the case that the set of rules of some $1$-cycle $\sigma$ of $R$ is in $\Delta_2$ but not in $\Delta_1$, while $\sigma$ is active w.r.t. $rn^*(I^{\sigma})$.
%If the set of such $1$-cycles of $R$ is non-empty, and for all other $k$-cycles $\sigma'$ of $R$ either $\mathsf{Rule}(\sigma')$ is in $\Delta_1$ or $\sigma'$ is not active w.r.t. $rn^*(I^{\sigma'})$, then $R$ is in $1$-$\mathsf{safe}(\Phi_{\Delta_2})$ but is not in $1$-$\mathsf{safe}(\Phi_{\Delta_1})$.} {\color{red} (I don't understand.)}

%{\color{blue}
Let us see what happens for $k=1$. Let $R$ be a rule set and $\Delta_1$ and $\Delta_2$ be any pair of acyclicity conditions where $\Delta_1\subset \Delta_2$.
Then there may be some active $k$-cycles $\sigma$ based on $R$ such that $\mathsf{Rule}(\sigma)\in \Delta_2\setminus \Delta_1$. If for all such cycles based on $R$ the above condition holds, then $R$ is in $1$-$\mathsf{safe}(\Phi_{\Delta_2})$ but not in $1$-$\mathsf{safe}(\Phi_{\Delta_1})$.
Hence, for different columns we see some differences between the numbers added to the first row of Table \ref{table:1} in a way that for any pair of acyclicity conditions $\Delta_1$ and $\Delta_2$ such that $\Delta_1\subset \Delta_2$, more rules are added as terminating to $k$-$\mathsf{safe}(\Phi_{\Delta_2})$ compared to $k$-$\mathsf{safe}(\Phi_{\Delta_1})$, when $k$ increases from $0$ to $1$.
%}
%we have the following relationship between their corresponding numbers of terminating rule sets: 
%more rules are added as terminating in 
%more $k$-cycles need to be evaluated for activeness when checking membership in $k$-$\mathsf{safe}(\Phi_{\Delta_1})$ compared to the case of $k$-$\mathsf{safe}(\Phi_{\Delta_2})$. The reason is that those $k$-cycles 
%Also, recall that given some cycle function $\Phi_\Delta$, when $k$ is an integer, a rule set $R$ is in $k$-$\mathsf{safe}(\Phi_\Delta)$ if for each $k$-cycle $\sigma$ based on $R$, either (1) according to Definition \ref{def:cycle_function_acyclicity}, $\mathsf{Rule}(\sigma)$ satisfies the acyclicity condition under consideration ($\Delta$), or	(2) based on Definition \ref{def:k-safe(cyc)} and Proposition \ref{lem:arbitraryDB-renamedDB}, $\sigma$ is not active w.r.t. $rn^{\ast}(I^{\sigma})$.
%
%Therefore, given a rule set $R$, it may be the case that the set of rules of some $1$-cycle $\sigma$ of $R$ is in $\Delta_2$ but not in $\Delta_1$, while $\sigma$ is active w.r.t. $rn^*(I^{\sigma})$.
%If the set of such $1$-cycles of $R$ is non-empty, and for all other $k$-cycles $\sigma'$ of $R$ either $\mathsf{Rule}(\sigma')$ is in $\Delta_1$ or $\sigma'$ is not active w.r.t. $rn^*(I^{\sigma'})$, then $R$ is in $1$-$\mathsf{safe}(\Phi_{\Delta_2})$ but is not in $1$-$\mathsf{safe}(\Phi_{\Delta_1})$.} {\color{red} (I don't understand.)}

%{\color{blue}
When $k>1$, we observe an interesting phenomenon  \-- the same number of terminating rule sets, in fact, the same rule sets, are added. Consider any pair of acyclicity conditions $\Delta_1$ and $\Delta_2$ such that $\Delta_1\subset\Delta_2$. For any rule set $R$, assume it is determined to be terminating by $\Delta_2$.
In general, more $k$-cycles based on $R$ need to be checked for the analysis of Algorithm \ref{alg:kSafe} in the case of $\Delta_1$ than $\Delta_2$, due to the weaker acyclicity condition in $\Delta_1$.
%The reason is that for any $k$-cycle $\sigma$ such that $\mathrm{Rule}(\sigma)\in \Delta_2\setminus \Delta_1$, $\sigma$ must be added to the collection of $k$-cycles to be checked for activeness in the case of $\Delta_1$, but it is removed from such activeness checking when we consider $\Delta_2$.
For each such $k$-cycle $\sigma$, 
since $\mathrm{Rule}(\sigma)\in \Delta_2$, it cannot be active w.r.t. $rn^{\ast}(I^{\sigma})$ for any renaming function $rn^{\ast}$ (otherwise it would contract Theorem \ref{key0}). 
Therefore, it must pass the activeness checking of Line 7 in Algorithm \ref{alg:kSafe}.
Consequently, just like how $R$ is determined to be terminating by $\Delta_2$, $R$ is determined to be terminating by $\Delta_1$. Conversely, since the set of $k$-cycles tested for $\Delta_1$ is a superset of those tested for $\Delta_2$, a rule set $R$ which is determined to be terminating by $\Delta_1$ must also be determined to be terminating by $\Delta_2$. This explains why the number of increases of terminating rule sets for different acyclicity conditions is a constant. 

This observation leads to a choice of strategy for testing for $k$-$\mathsf{safe}(\Phi_{\Delta})$ for an expensive acyclicity condition $\Delta$.  After failing the check of $1$-$\mathsf{safe}(\Phi_{\Delta})$, we can check 
 $k$-$\mathsf{safe}(\Phi_{\Delta'})$ for $k>1$, for a weaker by each-to-check acyclicity condition $\Delta'$ in the understanding that the same rule sets will be added, with likely more $k$-cycles to be tested.

Also, it is clear that if $\Delta_1\subset \Delta_2$ %(i.e., some skolem acyclicity condition $\Delta_1$ recognizes more rules as terminating compared to $\Delta_2$), 
then for all $k\ge 0$, $k$-$\mathsf{safe}(\Phi_{\Delta_1})\subset k$-$\mathsf{safe}(\Phi_{\Delta_2})$.
In Table \ref{table:1}, since $\WA\subset \JA\subset \text{MFA}$, the same inclusion relation holds for $k$-$\mathsf{safe}(\Phi_{\Delta})$ rules constructed from each acyclicity condition $\Delta$ for all integers $k\ge 0$.
%}

In order to compare our results with those of \cite{carral2017restricted}, we checked the set of terminating ontologies under our conditions for membership in RMFA introduced therein.
As a result, it was observed that all the tested rule sets, except for $2$ of them already belong to RMFA.
In fact, those $2$ rule sets belong to $6$-$\mathsf{safe}(\Phi_{\text{MFA}})$.

Additionally, we checked the tested corpora for membership in RMFC (cf. \cite{carral2017restricted}). It was observed that $165$ (23.6\%) of the ontologies belong to RMFC (i.e., for each rule in this category, there is a database $I_0$ for which the restricted chase of $R$ and $I_0$ is infinite).
Based on the above results we find %out 
that the termination status of $19$ (2.71\%) ontologies in the collection is open (i.e., they do not belong to $6$-$\mathsf{safe}(\Phi_{\text{MFA}})$ or RMFC or RMFA).
We conducted our tests for membership in RMFA and RMFC using VLog \cite{carral2019vlog}.
%}
%}

For the second experiment, we performed time analysis for the tested ontologies for different cycle functions by fixing $k$ to $6$. The results are reported in Table \ref{table:2}, where the average running time, as well as the number of ontologies {\em terminating within the average running time} (abbreviated as T.W.A.T.) for that particular cycle function, are reported.
It can be seen that in all tested conditions more than half of the terminating ontologies can be determined within the average time. Note that the average times of the table are in seconds.

From our experiments we can see that there is no {\em one-number-fits-all} $k$ for which any ontology belongs to $k$-$\mathsf{safe}(\Phi_{\Delta})$. However, as observed in our experiments, for real-world ontologies, this number can be indeed small.

\noindent
\textbf{TGD generator:} For adequately evaluating our approach and also for scalability testing, we implemented a TGD generator on top of \cite{benedikt2017benchmarking}. Our goal was to check the performance of implemented classes on large instances with large sets of chase atoms. 
Our generator can generate custom TGDs while controlling their complexity. It supports an arbitrary number of body atoms. Also, a parameterized total number of predicates and arity of atoms is defined. 
Furthermore, a parameter is used to control the maximum number of repeated relations in the formula.
Each TGD is generated by creating conjunctions and then selecting the subset of atoms that form the head of each TGD.

In our experiments, we generated $500$ linear source-to-target TGDs, and $200$ linear target TGDs. The reason that we picked linear rules was to control one parameter at a time and also to take the complexities of membership checking under control by focusing on the head atoms to have a better analysis on the restricted chase, as checking activeness of paths is the key here.

For the generated scenarios we precomputed restricted critical databases for the source instance generated by our TGD generator and then, to manage the structure of our TGDs, we tested $2$ different forms of TGD heads: 1) those that have three relations joined in a chain (i.e., the last variable of an atom is joined with the first variable of the next atom which we refer to as {\em chained TGDs}); 2) those in which three relations of the head do not share variables (which we refer to as {\em discrete TGDs}).

In all experiments, each atom has arity $4$ and each TGD can have up to $3$ repeated relations. The $3$ head predicates and the body predicate have been chosen randomly out of a space of $20$ predicates.
After the generation of each TGD, we check its membership in $k$-$\mathsf{safe}(\Phi_{\WA})$ for $k=\{0,1,2\}$; keep only those TGDs for which this test returns true and discard the rest. Results of different properties in the tested TGDs have been recorded in Tables \ref{table:3} and \ref{table:4}.

\begin{table}[!hbt]
	\centering
	\begin{tabular}{| c c c c c |}
		\cline{1-5}
		\multicolumn{5}{| c |}
		{\multirow{2}{*}
				{\textbf{Statistics of chained TGD generator for membership in $k$-$\mathsf{safe}(\Phi_{\WA})$}}}\\
				& & & & \\
		\cline{1-5}
		$k$ &
		Avg. time (s) &
		Terminating (\%) &
		Timeout failure (\%) &
		Memory failure (\%)
		\\
		\cline{1-5}
		$k=0$ &
		54 &
		81
		 & 
		 4 & 
		 0 
		 \\
		\cline{1-5}
		$k=1$ & 
		135 &
		83.3
		 & 
		6 & 
		0.2\\
		\cline{1-5}
		$k=2$ &
		474  &
		86.2
		& 7 &
		0.3\\
		\cline{1-5}
	\end{tabular}
	\vspace{2mm}
	\caption{Statistical results of chained TGD generator for $k$-$\mathsf{safe}(\Phi_{\WA})$ membership
	}
	\label{table:3}
\end{table}

\begin{table}[!hbt]
	\centering
	\begin{tabular}{| c c c c c |}
		\cline{1-5}
		\multicolumn{5}{| c |}
		{\multirow{2}{*}
		{\textbf{Statistics of discrete TGD generator for membership in $k$-$\mathsf{safe}(\Phi_{\WA})$}}}\\
		& & & & \\	
		\cline{1-5}
		$k$ & 
		Avg. time (s) &
		Terminating (\%) &
		Timeout failure (\%) &
		Memory failure (\%)
		\\
		\cline{1-5}
		$k=0$ & 
		43 &
		89 & 
		2 & 
		0 \\
		\cline{1-5}
		$k=1$ & 
		94 &
		92 &  
		4 & 
		0.11\\
		\cline{1-5}
		$k=2$ & 
		341  &
		92  &
		4.2  & 
		0.16\\
		\cline{1-5}
	\end{tabular}
	\caption{Statistical results of discrete TGD generator for $k$-$\mathsf{safe}(\Phi_{\WA})$ membership
	}
	\label{table:4}
\end{table}

The results of Tables \ref{table:3} and \ref{table:4} demonstrate that the average running times of membership checking in $k$-$\mathsf{safe}(\Phi_{\WA})$ for chained TGD generator is more than that of discrete TGD generator. The reason could be in the activeness checking module which takes more time in the rules in which (derived) atoms share variables. In addition, in both TGD generators, there are far lesser memory failures than timeout failures.

We performed the same check as detailed in the previous subsection regarding the need for utilizing updated restricted critical databases for rules outputted from our TGD generator, and similar to ontologies in the considered corpora in that section, we did not find any ontology for which we need to run experiments to check activeness with updated restricted critical databases.

%{\color{blue}
\section{Discussion}
\label{8}
%{\color{blue}
In this section, we introduce some more recent papers in this area and then show how to leverage them to extend our proposed $k$-$\mathsf{safe}(\Phi)$ classes uniformly.
%}
%{\color{red} (Say what is the goal of this section. Discussing what?)} 
In \cite{krotzsch2019power} it is shown that there are examples of TGDs for which the data complexity of the restricted chase can reach non-elementary upper bounds. Note, however, that as shown in \cite{zhang2015existential}, given any $\kappa>0$, for any $\mathsf{exp}_{\kappa}$-bounded rule set $R$ under the skolem chase variant, the Boolean query answering problem is \textsc{PTime}-complete for the data complexity. Therefore, the restricted chase can {\em realize} queries which are out of the reach for the skolem chase variant.

Let us define the notion of a {\em strategy} as a plan of choosing paths based on a given rule set. 
Utilizing this notion allows us to focus on a {\em concrete plan for path selection} in the course of our termination analysis for the restricted chase to extend the set of terminating rules under the restricted chase.
Exploiting the above terminology, $\mathsf{CT}_{\forall\forall}^\mathsf{res}$ can be alternatively defined to be the set of rules with {\em terminating restricted chase for all strategies} and all instances.

On the other hand, from \cite{onet2013chase} it is known that $\mathsf{CT}_{\forall\forall}^\mathsf{res}\subset \mathsf{CT}_{\forall\exists}^\mathsf{res}$, where $\mathsf{CT}_{\forall\exists}^\mathsf{res}$ denotes the class of rule sets $R$ such that for all instances $I$ there exists at least one restricted chase sequence of $I$ and $R$ that is finite.
Similarly, we can define $\mathsf{CT}_{\forall\exists}^\mathsf{res}$ to be the set of rules with {\em terminating restricted chase for some strategy} and all instances.

Recently, a chase variant known as the {\em Datalog-first chase} has been introduced in \cite{carral2017restricted} and subsequently in \cite{krotzsch2019power}, which extends all-path restricted chase by focusing on a particular class of strategies that priorities the application of non-generating (Datalog) rules in any considered restricted chase sequence.
Let $\mathsf{CT}_{\forall\forall}^\mathsf{dlf}$ denote the set of rules with a terminating Datalog-first 
chase for all strategies (paths) and all instances.\footnote{Note that by strategy in the Datalog-first chase we mean a plan for choosing the application order of Datalog rules which must always occur before the application of generating rules (i.e., non-full TGDs).}
Then we have $\mathsf{CT}_{\forall\forall}^\mathsf{res}\subset\mathsf{CT}_{\forall\forall}^\mathsf{dlf}\subseteq \mathsf{CT}_{\forall\exists}^\mathsf{res}$.
Note that although the first inclusion is strict, at the time of writing this paper it is not known whether the second inclusion above is also strict or not.

Based on what was discussed above, we can extend the set of $\delta$-bounded rules under the restricted chase variant as well as $k$-$\mathsf{safe}(\Phi)$ rules for a given bound function $\delta$, integer $k$ and cycle function $\Phi$ by considering the Datalog-first chase.
For this purpose, we only need to focus on cycles in which Datalog rules appear before non-Datalog rules (except for the last rule of each path). 
Given a bound function $\delta$, let us call any rule set $R$ that is $\delta$-bounded under the above condition {\em $\delta$-bounded under the Datalog-first chase}. We can define {\em $k$-$\mathsf{safe}(\Phi)$ rules under the Datalog-first chase} similarly.\footnote{In this case, $R$ is said to be in $k$-$\mathsf{safe}(\Phi)$ under the Datalog-first chase, or to belong to $k$-$\mathsf{safe}(\Phi)$ under the Datalog-first chase (given a cycle function $\Phi$ and an integer $k$), if for every $k$-cycle $\sigma$ which prioritises Datalog rules (except the last rule of the cycle), and is mapped to $F$ under $\Phi^R$, $\sigma$ is safe.}

\begin{exm} \label{exm:fairness}
Let $R=\{r_1,r_2\}$ (adopted from \cite{gogacz2019all}), where
	$$
	\begin{array}{ll}
	r_1\!:  ~ Q(x,y,y) \rightarrow \exists u Q(x,u,y), Q(u,y,y)\\
	r_2\!:  ~ Q(x,y,z) \rightarrow Q(z,z,z)
	\end{array}
	$$

Note that in this rule set the fairness condition requires application of $r_2$ in any (fair) sequence of the restricted chase and after $r_2$ is applied, the next application of $r_1$ is not active and therefore, any (fair) restricted chase sequence terminates. The following derivation starting from $\{Q(a,b,b)\}$ demonstrates such a sequence in which fresh nulls $z_i$ are used to instantiate the existential variable $u$:
	$$
	\begin{array}{ll}
	I_0 = \big\{Q(a,b,b)\big\} \xlongrightarrow{\langle r_1,\{x/a, y/b\}\rangle}\\
	I_1 = I_0 \cup \big\{Q(a,z_1,b),Q(z_1,b,b)\big\} \xlongrightarrow{\langle r_1, \{x/z_1, y/b\}\rangle}\\
	I_2 = I_1 \cup \big\{Q(z_1,z_2,b),Q(z_2,b,b)\big\}
	\xlongrightarrow{\langle r_1, \{x/z_2, y/b\}\rangle}\\
	I_3 = I_2 \cup \big\{Q(z_2,z_3,b),Q(z_3,b,b)\}\xlongrightarrow{\langle r_1, \{x/z_3, y/b\}\rangle}\\
	\dots\\
	%\xlongrightarrow{\langle r_1, \{x/z_{j-2}, y/b\}\rangle}
	I_{j-1} = I_{j-2} \cup \big\{Q(z_{j-2},z_{j-1},b),Q(z_{j-1},b,b)\}
	\xlongrightarrow{\langle r_2, \{x/a, y/b, z/b\}\rangle}\\
	I_j = I_{j-1} \cup \big\{Q(b,b,b)\}\\
	%\xlongrightarrow{\langle sk(r_1), \{x/z_{j-1}, y/b\}\rangle}\\
	%I_{j+1} = I_j \cup \big\{Q(z_{j-1},z_j,b),Q(z_j,b,b)\}\\
	\end{array}
	\vspace{.1in}
	$$

Note that in the above sequence of derivations, the following step: $I_j\langle r_1, \{x/z_{j-1}, y/b\}\rangle I_{j+1}$ does not exist, and any valid restricted chase sequence terminates. However, the fairness condition needs the existence of some $j$ to apply some active trigger involving $r_2$ (i.e., $\langle r_2, \{x/a, y/b\}\rangle$ in this example). But due to the non-deterministic nature of this process, $j$ can be chosen anywhere in the sequence.

As discussed above, rule set $R$ in this example is not $\delta$-bounded under the restricted chase for any computable bound function $\delta$. However, starting from any database $I$, no (fair) infinite restricted chase sequence can be constructed from $R$ and $I$.

On the other hand, it is not hard to see that $R$ belongs to $1$-$\mathsf{safe}(\Phi_\text{WA})$ under the Datalog-first chase. The reason is that this chase variant requires the application of $r_2$ before $r_1$ in any valid chase sequence. Therefore, all $1$-cycles in which the application of Datalog rules are prioritised (i.e., $(r_2,r_1,r_2)$) are safe. The following sequence of derivations shows why this is the case.
	$$
	\begin{array}{ll}
	I'_0 = \big\{Q(a,b,c)\big\} \xlongrightarrow{\langle r_2,\{x/a, y/b, z/c\}\rangle}\\
	I'_1 = I'_0 \cup \big\{Q(b,b,b)\big\} \xlongrightarrow{\langle r_1, \{x/b, y/b\}\rangle}\\
	I'_2 = I'_1 \cup \big\{Q(b,z_1,b),Q(z_1,b,b)\big\}
	\xLongrightarrow{\theta=\{z_1/b\}} \theta(I'_2)\subseteq I'_1
	\end{array}
	\vspace{.1in}
	$$
\end{exm}

%Notice that the fairness problem does not show up for the case of the skolem chase variant.
%The reason is that as shown in \cite{calautti2015chase}, for this chase variant the existence of one (unfair) infinite chase derivation of $I$ w.r.t. a rule set $R$ implies the existence of a fair one.

Notice that as recently shown in \cite{gogacz2019all}, if the given rule set is single-head (i.e., all rules in it are single-head), then the fairness for the restricted chase termination is irrelevant. 
However, unlike the skolem chase variant for which there is a straightforward termination-preserving translation from any rule set to a single-head rule set (cf. \cite{baget2011rules}), no such termination-preserving translation exists for the restricted chase.

Clearly, given an integer $k$ and a cycle function $\Phi$, any rule set that is $k$-$\mathsf{safe}(\Phi)$ under the restricted chase is also $k$-$\mathsf{safe}(\Phi)$ under the Datalog-first chase. 
Example \ref{exm:fairness} shows that this inclusion relation is indeed strict. 
The same argument holds for $\delta$-bounded rules under the restricted vs. the Datalog-first chase using the same example to demonstrate that the inclusion is strict.
%}

\section{Conclusion and Future Work}\label{conclusion}
	In this work, we introduced a general framework to extend classes of chase terminating rule sets. We formulated a technique to characterize finite restricted chase which can be applied to extend any class of finite skolem chase identified by a condition of acyclicity. The main strength of our work, which is also the main distinction from almost all previous work on chase termination, is its generality.
	Then, we showed how to apply our techniques to extend $\delta$-bounded rule sets.
	Our theoretical results for complexity analyses showed that in general this extension indeed increases the complexities of membership checking and the complexity of combined reasoning tasks for $\delta$-bounded rule sets under the restricted chase compared to the skolem chase.
	However, by implementation and experimentation, we showed the relevance of our work in real-world ontologies.  
	Our experimental results discovered a growing number of practical ontologies with finite restricted chase by increasing computational cost as well as changing the underlying cycle function. Our experimentation also showed evidence that existential rules provide a suitable modeling language for ontological reasoning.

We will next investigate conditions for subclasses with a reduction of cost for membership testing. One idea is to find syntactic conditions under which triggers to a rule are necessarily active.

The current implementation of our system is relatively slow, particularly for non-linear rules. It often requires long chase times to check membership in $k$-$\mathsf{safe}(\Phi)$ even for small values of $k$. 
In order to tackle this problem, we can consider two options which can also be combined towards a more efficient implementation. 
%{\color{blue}
The first option is considering $k$-$\mathsf{safe}(\Phi)$ under a particular path selection strategy such as the Datalog-first approach. This way, we can filter out a subset of paths in our membership analysis.
Alternatively, we can conduct our implementations in {\em MapReduce} model.

\bibliographystyle{acmtrans}
\bibliography{bibfile_ext}

\label{lastpage}
\end{document}